\documentclass[12pt]{article}
\usepackage{hyperref}
\usepackage{float}
\usepackage{amsmath}
\usepackage{graphicx,psfrag,epsf}
\usepackage{enumerate}
\usepackage{natbib}
\usepackage{multirow}
\usepackage[ruled,vlined]{algorithm2e}
\usepackage{array}
\usepackage{graphicx}
\usepackage[utf8 ]{inputenc}
\usepackage[english]{babel}
\usepackage{amsthm}
\usepackage[bottom]{footmisc}
\usepackage{longtable}
\usepackage{mathrsfs}
\usepackage{extarrows}
\usepackage{bigints}
\usepackage{subcaption}
\usepackage{todonotes}
\usepackage{hyperref}
\usepackage{algorithm2e}
\usepackage{setspace}
\usepackage{bbm}
\usepackage[top=1in, bottom=1in, left=1in, right=1in]{geometry}
\usepackage{amsmath}
\usepackage{amsfonts}

\newtheorem{theorem}{Theorem}[section]

\newtheorem{assumption}{Assumption}

\def\bSig\mathbf{\Sigma}

\def\mb#1{\setbox0=\hbox{$#1$}
  \kern-.025em\copy0\kern-\wd0
  \kern.05em\copy0\kern-\wd0
  \kern-.025em\raise.0em\box0}

\newcommand{\floor}[1]{\lfloor #1 \rfloor}

\DeclareMathOperator*{\argmax}{arg\,max}

\begin{document}

    \title{\vspace{-2cm}Frequency Band Analysis of Nonstationary Multivariate Time Series 
}

\author{Raanju R. Sundararajan\footnote{Email: \href{mailto:rsundararajan@smu.edu}{rsundararajan@smu.edu} (both authors contributed equally to this work)} \, \, and \, \, Scott A. Bruce\footnote{Email: \href{mailto:sabruce@tamu.edu}{sabruce@tamu.edu} } \\
	$^\ddag${\it \small Department of Statistical Science, Southern Methodist University,  Dallas, Texas, USA} \\
	$^\S${\it \small Department of Statistics, Texas A\&M University, College Station, Texas, USA}}
 
\date{}

\maketitle

\begin{abstract}

\noindent Information from frequency bands in biomedical time series provides useful summaries of the observed signal. Many existing methods consider summaries of the time series obtained over a few well-known, pre-defined frequency bands of interest. However, these methods do not provide data-driven methods for identifying frequency bands that optimally summarize frequency-domain information in the time series. A new method to identify partition points in the frequency space of a multivariate locally stationary time series is proposed. 
These partition points signify changes across frequencies in the time-varying behavior of the signal and provide frequency band summary measures  
that best preserve the nonstationary dynamics of the observed series. 
An $L_2$ norm-based discrepancy measure that finds differences in the time-varying spectral density matrix is constructed, and its asymptotic properties are derived. New nonparametric bootstrap tests are also provided to identify  significant frequency partition points and to identify components and cross-components of the spectral matrix exhibiting changes over frequencies. Finite-sample performance of the proposed method is illustrated via simulations. 
The proposed method is used to develop optimal frequency band summary measures for characterizing time-varying behavior in resting-state electroencephalography (EEG) time series, as well as identifying components and cross-components associated with each frequency partition point.
  
\end{abstract}

\hrule
\hrulefill

\section{Introduction}
\label{sec:intro}
Frequency-domain information contained in biomedical time series provides important summaries leading to meaningful physiological interpretations.  Oscillatory behavior 
in nonstationary time series are often characterized by the time-varying spectral matrix, which is a time- and frequency-dependent matrix-valued function. Rather than analyzing the full time-varying spectral matrix, practitioners often partition frequencies into a few
well-known, pre-defined bands that serve as a pseudo partition of the frequency space. 
Numerical summaries of the spectral matrix are then computed within these frequency bands, and these summaries have been shown to provide meaningful biological interpretations of the observed signal. Examples of biomedical time series where frequency bands are identified and associated with physiological characteristics include heart rate variability (HRV) (\citealp{Halletal2004}), local field potential (LFP) (\citealp{lfp-fband}), resting-state functional magnetic resonance imaging (rs-fMRI) (\citealp{rsfmri-fband}) and electroencephalography (EEG) (\citealp{Klimesch1999}). 

In analyzing neuroimaging data, several works have identified frequency bands of interest that carry meaningful biological interpretations. \citet{biswal-rsfmri-lowfreq} analyze correlations in resting state fMRI that characterize functional connectivity in the brain. In their study, high correlations in low frequency oscillations ($<0.1$ Hz) are detected within the sensorimotor cortex of the brain and also with other regions of the brain that associate with motor function. \citet{lowe-rsfmri-lowfreq} is another work that computes cross-correlations in low frequency oscillations ($<0.08$ Hz) of the fMRI signal between widely separated anatomic regions of the brain. To detect regions of the brain that exhibit strong cross-correlations, their work uses these low frequency oscillations to discriminate between different memory tasks. In \cite{rsfmri-fband}, frequency partitions within the 0.01-0.25 Hz band were obtained for rs-fMRI signals and the oscillations from each of these partitions were associated with different biological functions such as respiration, pulse, and vasomotor activity.  In \cite{lfp-lowfreq}, spectral analysis is performed on a LFP signal gathered from the primary visual cortex of a macaque. The power of the spectral density of this LFP signal is computed at the low ($< 10$ Hz), gamma ($25-240$ Hz) and broad ($8-240$ Hz) frequency bands, and changes in power are analyzed in response to various stimuli. In analyzing EEG data, many methods resort to analyzing oscillations for known frequency bands, such as Alpha, Beta, Gamma and Delta, and associate these oscillations with various biological functions. As an example,  \citet{eeg-freqbands-review-paper} provides a review of several methods that utilize pre-defined and well-known frequency bands to analyze resting state EEG data from individuals with neurological disorders. In all of the above cited works, irrespective of the modality of choice (fMRI, LFP or EEG), the selection of frequency bands used to summarize oscillatory patterns is often predetermined either through manual inspection of the time series or through historical precedents. \cite{Doppelmayretal1998} mention the possibility of differences in the endpoints of the Alpha frequency band in EEG data from different individuals, and  \cite{eeg-peak-alpha-freq-changes} is another example that discusses differences in the peak frequency in the Alpha band of EEG data among individuals from different sexes.  Analyses like these point to the need for a data-driven method for automatically identifying frequency bands that best describe changes in the frequency space of neuroimaging time series data.   

Several methods in the signal processing and applied harmonic analysis literature pursue time-frequency analysis of univariate nonstationary time series. Time-frequency analysis helps understand the time-varying properties of different frequency components of the time series. For obtaining a time-frequency description of the observed signal, \citet{Flandrin1998} describes a few solutions, such as the short-time Fourier transform and wavelet transform, along with their properties. While most methods aim at estimating the time-varying amplitudes of the signal, very few methods discuss estimation of changes happening in the frequency space of nonstationary time series. An online algorithm for detecting a single prominent time-varying frequency band is proposed in \citet{Tiganjetal2012}. In their work, at each local time window, the prominent frequency band is estimated by using a band-limited signal as the input. In  \citet{Cohenetal2016}, the observed signal is assumed to be composed of multiple, uncorrelated cyclostationary processes \citep{Gardneretal2006}, and they provide an algorithm for detecting multiple peaks in the spectral density of the signal.

In the statistical literature, time-frequency analysis of nonstationary time series has been studied using evolutionary spectra (\citealp{Priestley1965}) and  the time-varying spectral density of locally stationary processes (\citealp{Dahlhaus1997}). Most methods utilize the time-varying spectral matrix or its finite sample estimates for detecting changes in the time space, i.e., temporal change point detection; see \citet{Adak1998}, \citet{Ombaoetal2005}, \citet{Kirchetal2015}, \citet{Preussetal2015} for examples in univariate and multivariate locally stationary time series. In \citet{SchroderandOmbao2016}, temporal change points are obtained over user-specified frequency bands and the obvious limitation in their work is the need for the user to specify the frequency bands over which the temporal change points are found. \citet{bruce_etal_2020} is a recent work that considers univariate locally stationary time series and aims at finding partition points in the frequency space that best describe the nonstationary characteristics of the observed time series. Their work constructs an estimate of the time-varying spectral density by assuming a piecewise stationary approximation and applying a multitaper estimator of the spectral density. A CUSUM statistic designed to spot changes in the frequency space is formulated based on this estimator, and significant partition points are obtained.   

\noindent In this work, we propose a new method to identify partition points in the frequency space of a multivariate locally stationary time series. These partition points are detected such that they best preserve the nonstationary dynamics of the time-varying spectral density matrix. To detect these partition points, an $L_2$ norm-based discrepancy measure is constructed using the time-varying spectral matrix. This measure computes differences in the time-varying spectral matrix over local neighborhoods of frequencies and its asymptotic properties are derived. With the discrepancy measure viewed as our test statistic, a new nonparametric bootstrap test is proposed to obtain the critical values. The proposed bootstrap procedure is also further utilized in identifying the components and cross-components contributing to the changes in the frequency space. The proposed method is motivated by an application in analyzing resting-state electroencephalography (EEG) time series from 14 individuals. The experiment involves a straightforward eyes-open/eyes-closed scheme with the signal being recorded from 16 electrodes \citep{gregoire_cattan_2018_2348892}. After down-sampling the original signal to the rate of 64 Hz, the left plots in Figure \ref{fig:intro_figure} depicts the standardized EEG time series from the Oz occipital channel in participants 2 and 13. These illustrated signals cover a time segment of the entire experiment that includes a total of five consecutive blocks, with each block being an eyes-closed or eyes-open scenario. The right plots in the same figure show the estimated time-varying spectral densities for this channel, and the time-varying behavior can be witnessed around the 10 Hz frequency. 

This work addresses several critical aspects of frequency-domain analysis of nonstationary multivariate signals.  First, a data-driven estimator for frequency bands of interest for each individual is developed, as opposed to investigating commonly known pre-defined frequency bands that are assumed to be the same for all individuals.  Second, the proposed method is used to better understand variability in estimated partition points across different individuals from a population of interest.  Finally, it is typically not the case that all components of the spectral matrix exhibit changes across frequencies for each partition point.  Accordingly, the proposed method allows for detection of the particular components and cross-components of the time series that are contributing to each estimated frequency partition point. With the proposed method being uniquely positioned to address such problems, in Section \ref{sec:application} we illustrate and discuss the application in detail.

\begin{figure}[H]
    \centering
   \begin{subfigure}[b]{0.48\textwidth}
     \includegraphics[width=1\textwidth]{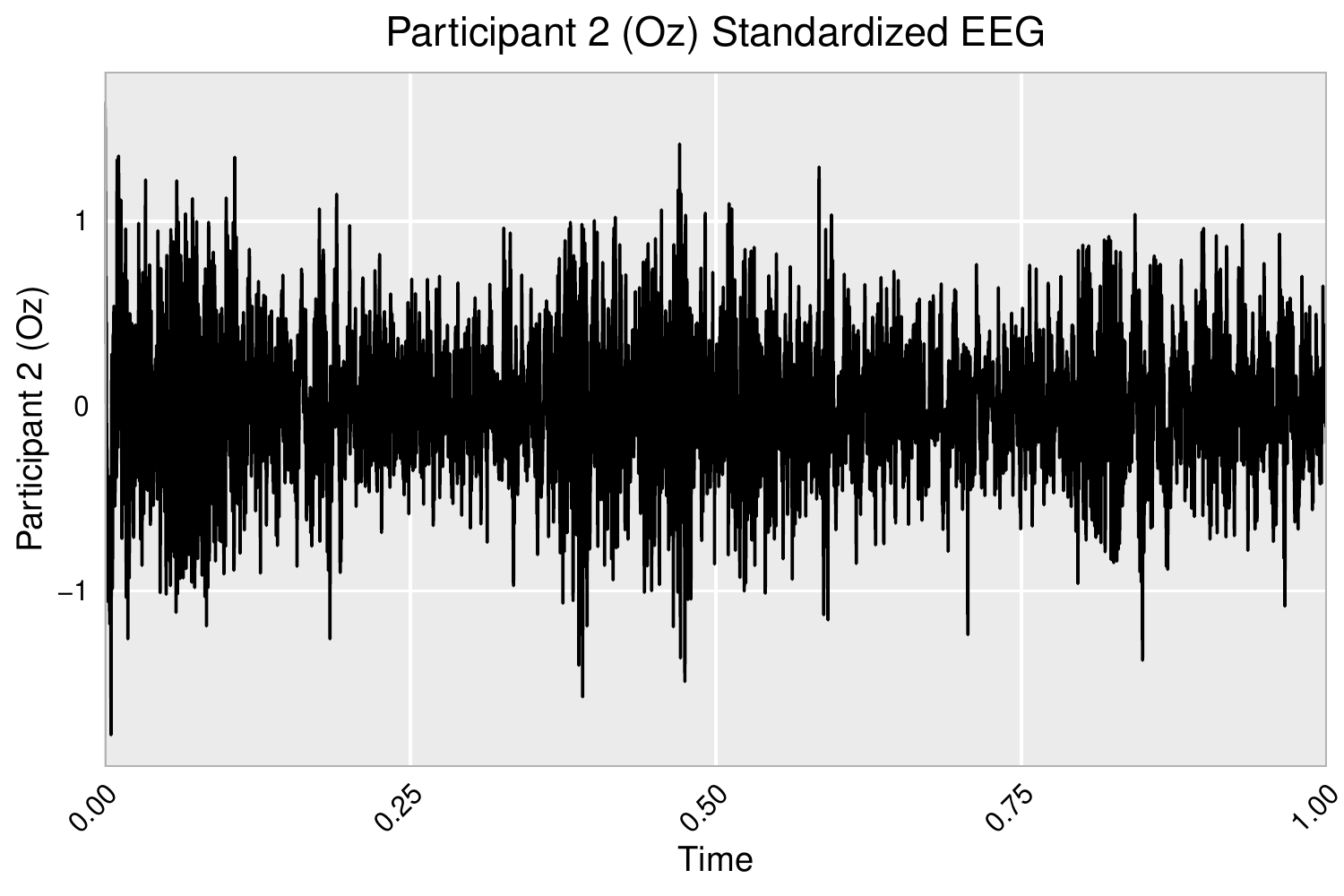}
    \caption{}
   \end{subfigure} 
     \begin{subfigure}[b]{0.48\textwidth}
     \includegraphics[width=\textwidth,trim={0.3cm 1cm 1cm 0.75cm},clip]{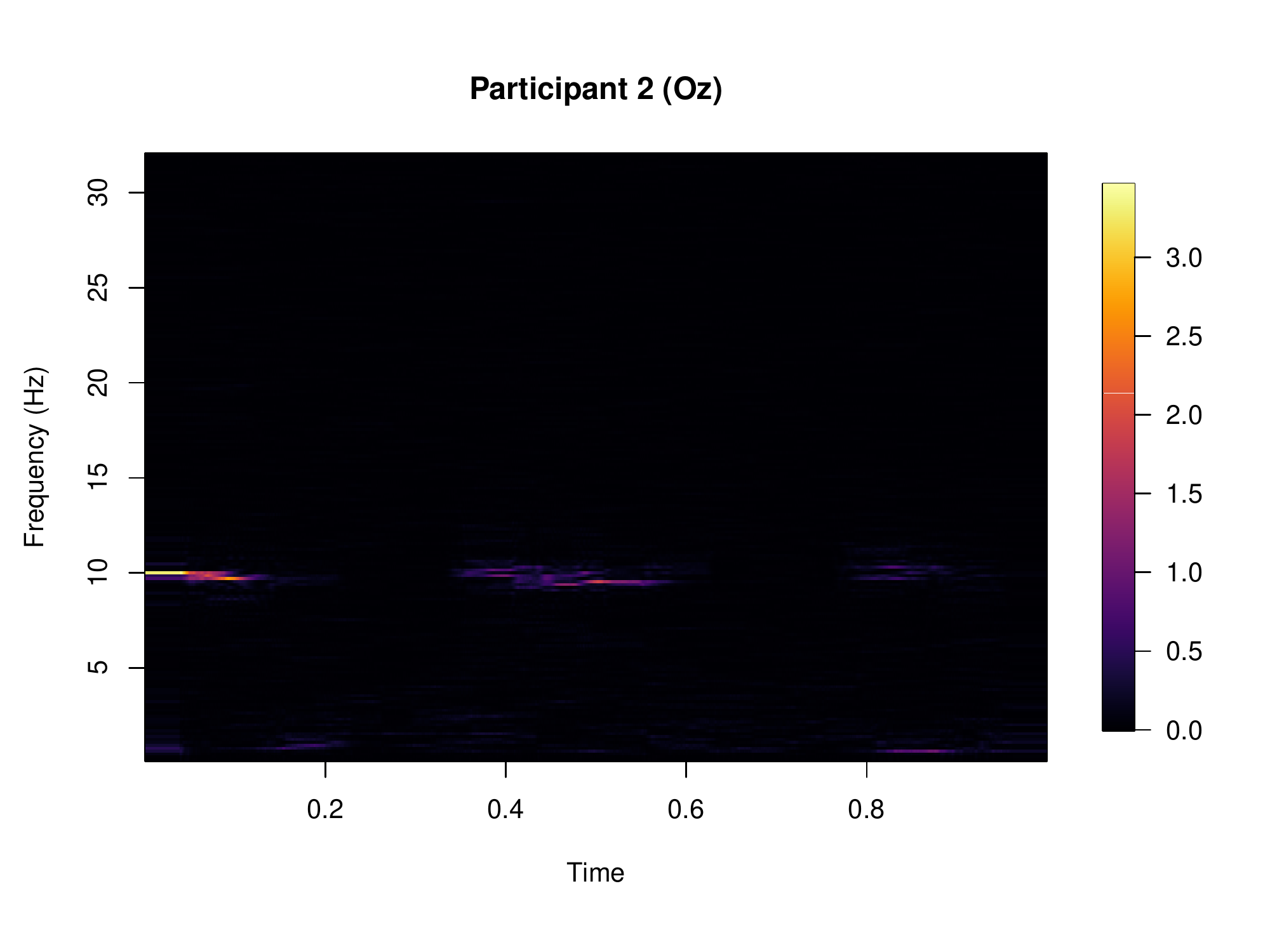}
    \caption{}
   \end{subfigure} 
      \begin{subfigure}[b]{0.48\textwidth}
     \includegraphics[width=1\textwidth]{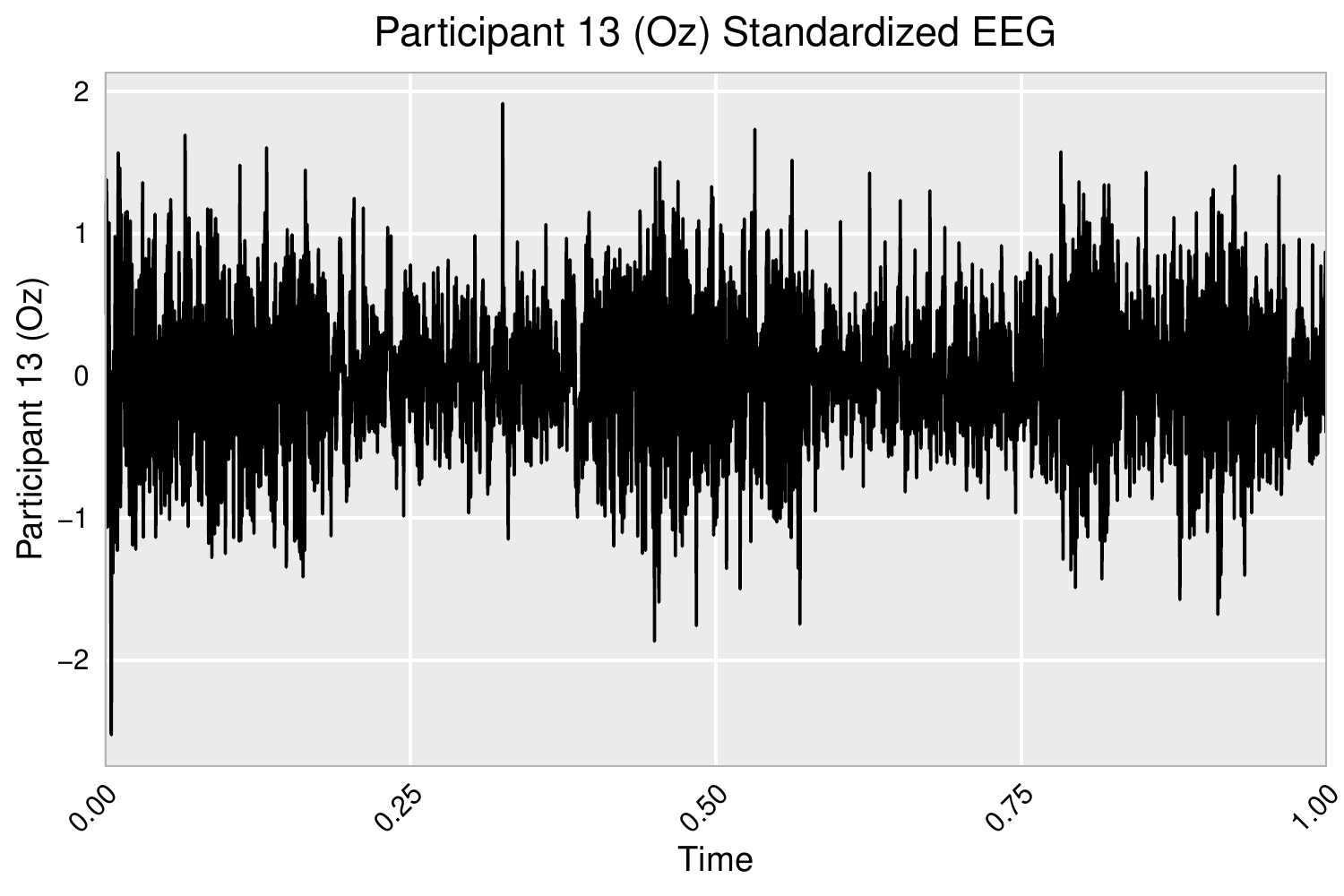}
    \caption{}
   \end{subfigure} 
     \begin{subfigure}[b]{0.48\textwidth}
     \includegraphics[width=\textwidth,trim={0.3cm 1cm 1cm 0.75cm},clip]{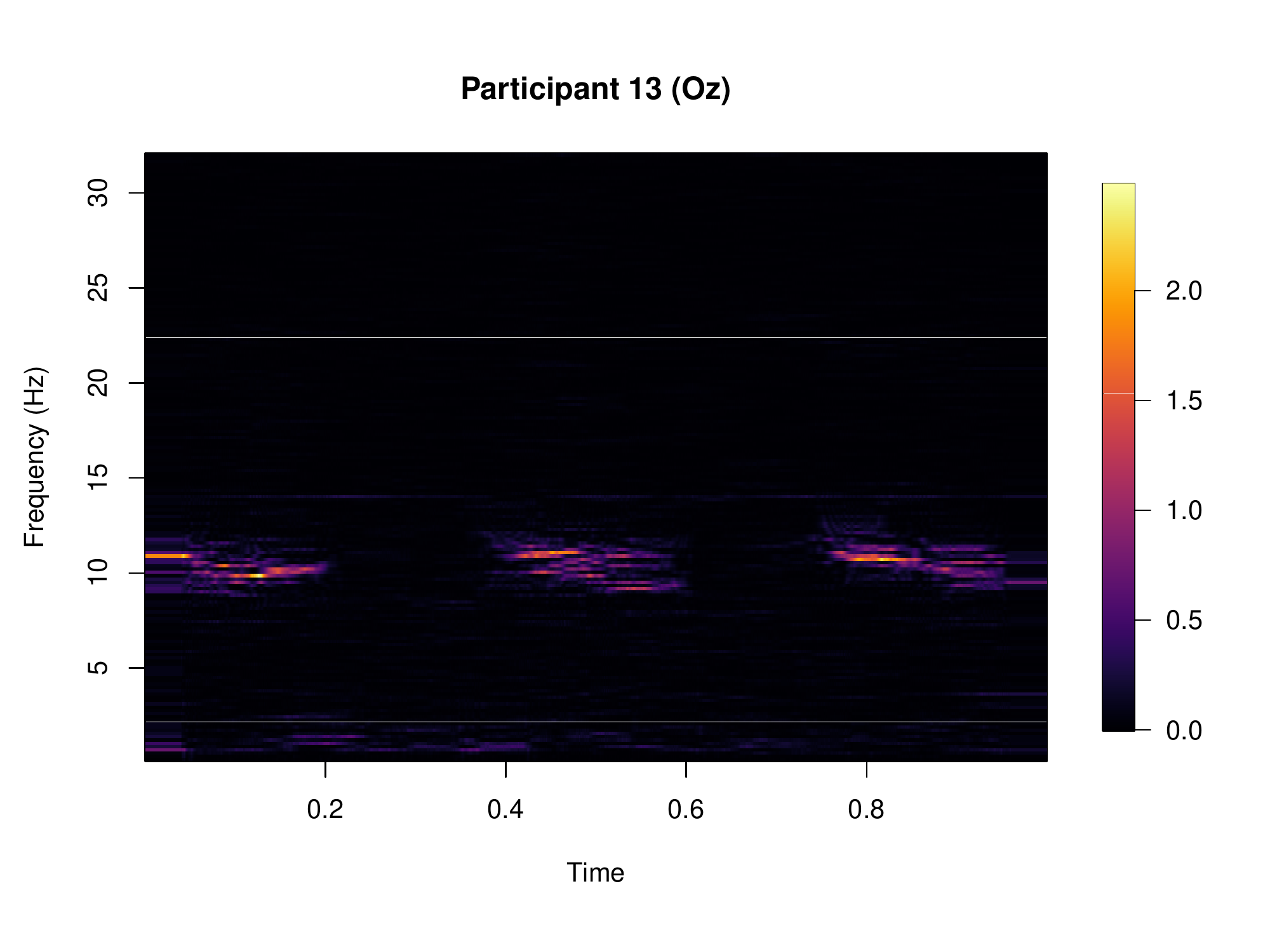}
    \caption{}
   \end{subfigure} 
    \caption{Oz channel EEG for participants 2 and 13 and corresponding estimated spectral densities.}
    \label{fig:intro_figure}
\end{figure}

\noindent The rest of the paper is organized as follows. Section \ref{sec:methodology} describes the proposed method in detail along with the theoretical results. Section \ref{sec:simulations} discusses the finite sample performance of the method using a few simulation schemes. The application to modeling resting-state EEG time series data is presented in Section \ref{sec:application}. The concluding remarks are given in Section \ref{s:conclusion}.

\section{Methodology}
\label{sec:methodology}

In this section, we describe our proposed method to find frequency partition points in the time-varying spectral matrix of a locally stationary process.  First, the model for the locally stationary process and frequency-banded time-varying behavior is introduced.  A discrepancy measure and estimator is then proposed for identifying frequency partition points that characterize the frequency band structure.  Resampling-based testing procedures are then developed to determine the significance of potential frequency partition points and to identify components associated with significant frequency partition points.


\subsection{Model}
\label{sec:model}
We begin with the definition of a locally stationary process using a two-sided MA$(\infty)$ representation followed by the required assumptions on the coefficient matrices (\citealp{Dahlhaus1997}, \citealp{dahlhaus2000}).  
Let $X_{t,T}$ be a $p$-variate locally stationary process given by 
\begin{equation}\label{eq:locally_stationary_ts}
X_{t,T} = \sum_{j = -\infty}^{\infty} \Phi_{t,T,j} \varepsilon_{t-j}, \;\;\; t=1,2,\hdots,T,
\end{equation}
where $\varepsilon_{t}$ are i.i.d Gaussian with unit variance matrix. The time-varying coefficient matrices, $\Phi_{t,T,j}$, are assumed to be temporally smooth in the following sense.

\begin{assumption}[Temporal smoothness]
\label{as:smooth}
There exists temporally smooth functions $\Phi: \lbrack 0,1 \rbrack \times \mathbb{Z} \rightarrow \mathbb{R}^{p \times p}$ such that
$$
\sum_{j = -\infty}^{\infty} \sup_{t=1,2,\hdots,T} \left|\left| \Phi_{t,T,j} - \Phi\left(t/T , j\right) \right|\right|_{\infty} = O\left(1/T\right), 
$$
where $||\cdot ||_{\infty}$ denotes the infinity norm, and
$$\sum_{j = -\infty}^{\infty} \sup_{u \in \lbrack 0,1 \rbrack } || \Phi(u , j) ||_{\infty} |j| < \infty,
$$
$$
\sum_{j = -\infty}^{\infty} \sup_{u \in \lbrack 0,1 \rbrack } || \Phi^{'}(u , j) ||_{\infty} |j| < \infty,
$$
$$\sum_{j = -\infty}^{\infty} \sup_{u \in \lbrack 0,1 \rbrack } || \Phi^{''}(u , j) ||_{\infty}  < \infty.
$$\end{assumption}

 With the above model, the  $p \times p$ time-varying spectral matrix of $X_{t,T}$ is given by 
\begin{equation}\label{eq:tvspec}
f(u,\omega) = \frac{1}{2\pi} \sum_{r,s} \Phi(u,r) \Phi(u,s)^{'} e^{-i 2 \pi \omega (r-s)}, 
\end{equation}

 for $\omega \in \lbrack -1/2 ,1/2 \rbrack$. With the above temporal smoothness assumption, the series $X_{t,T}$ can also be expressed using the Cram\'{e}r representation with a time-varying transfer function matrix that can be approximated by a temporally-smooth matrix-valued function $A: \lbrack 0,1 \rbrack \times \lbrack -1/2 ,1/2 \rbrack \rightarrow \mathbb{C}^{p \times p} $. Then, the time-varying spectral matrix can be written as $f(u,\omega) = A(u,\omega) A^{*} (u,\omega)$, where $A^{*}(u,\omega)$ denotes the conjugate transpose.  

 To characterize the nonstationary behavior of the spectral matrix beyond a simple level shift, we consider the demeaned time-varying spectral matrix given by
\begin{equation}\label{eq:demeaned_spectral_matrix}
g(u, \omega) = f(u ,  \omega) - \int_{0}^{1} f(u , \omega) du.
\end{equation} 

\noindent We will further assume that $g(u,\omega)$ has the following partition in the frequency space.

\begin{assumption}[Frequency-banded time-varying behavior]

\begin{equation*}\label{eq:g_frequency_partition}
g(u , \omega) = \begin{cases}
       g_1(u) &\quad\text{if} \;\; \omega \in [0,\omega_1) \\
       g_2(u) &\quad\text{if} \;\; \omega \in \lbrack \omega_1 ,\omega_2) \\
       \vdots \\
       g_{K+1}(u) &\quad\text{if} \;\; \omega \in \lbrack \omega_{K} , 1/2],        
       \end{cases}
\end{equation*} 
where $P_K = \{ \omega_1,\omega_2,\hdots, \omega_K \}$ denotes the set of $K$ frequency partition points in the interval $(0,1/2)$, with $K$ and $P_K$ being fixed and unknown.
\end{assumption}


\subsection{Comparing frequency-specific time-varying behavior}
\label{sec:teststatistic}
 To determine both $K$ and $P_K$, we consider the following $L_2$ norm-based discrepancy measure on the demeaned spectral matrix $g(u,\omega)$. For any $\omega \in (0,1/2)$ we have,
\begin{equation}\label{eq:main_discrepancy_measure}
D(\omega) = \frac{1}{\delta} \int_{u=0}^{1} \int_{\lambda = 0}^{\delta} \left|\left| g(u, \omega - \lambda) - g(u, \omega + \lambda) \right|\right|^2 \;  d\lambda \; du,
\end{equation}
where $|| \cdot ||^2$ denotes the squared $L_2$ norm, and $\delta > 0$. Observe that at any frequency $\omega$, the nonnegative measure $D(\cdot)$ represents the difference in the demeaned time-varying spectral matrix in a local neighborhood of frequencies surrounding $\omega$.

 In order to estimate the discrepancy measure in \eqref{eq:main_discrepancy_measure}, we use the local periodogram $I_N(u,\omega)$ (\citealp{Dahlhaus1997}) given by 
\begin{equation} \label{eq:local_periodogram}
I_N(u,\omega) =  J_N(u,\omega)J_N^{*}(u , \omega), \;\; \textrm{with} \;\; J_N(u , \omega) = \frac{1}{\sqrt{2 \pi N}} \sum_{s = 0}^{N-1} X_{ \floor{uT} - N/2 + 1 + s,T } \; e^{-i 2 \pi \omega s},
\end{equation}
where $N$ is the length of the local neighborhood around time point $u$ and $J_N^{*}(u,\omega)$ denotes the conjugate transpose. Viewing $I_N(u,\omega)$ as a local estimate of the time-varying spectral matrix $f(u,\omega)$, we obtain the estimated version of our measure in \eqref{eq:main_discrepancy_measure} as

\begin{equation}\label{eq:main_discrepance_measure_estimated}
\widehat{D}(\omega) = \frac{1}{T} \sum_{t=1}^T \frac{1}{W} \sum_{k=1}^{W} \left|\left| \widehat{g} \Big( t/T , \omega - \lambda_k \Big) - \widehat{g} \Big( t/T , \omega + \lambda_k \Big) \right|\right|^2, 
\end{equation}  
where $\widehat{g}( t/T , \omega - \lambda_k ) = I_N( t/T , \omega - \lambda_k ) - \frac{1}{T} \sum_{t_1=1}^{T} I_N( t_1/T , \omega - \lambda_k ) $, and $\lambda_k = k/N, \; k=1,2,\hdots,W$. $W$ corresponds to the parameter $\delta$ in \eqref{eq:main_discrepancy_measure}, and in finite sample situations, we resort to a multiscale approach in which we consider a range of plausible values for $W$.  This approach is discussed in Section \ref{sec:choice_of_W}. 

 To better understand the behavior of the measure $\widehat{D}(\cdot)$ in \eqref{eq:main_discrepance_measure_estimated}, we consider the following sets.
\begin{align} \label{eq:points_sets_defn}
\begin{split}
C_N &= \left\{ \frac{W}{N} , \frac{W+1}{N} , \hdots , \frac{1}{2} - \frac{W}{N}  \right\}, \\
C_{N,1} &= \bigcup_{j=1}^{K} \left\{ \omega_j - \frac{W}{N}, \hdots , \omega_j + \frac{W}{N} \right\}, \; \mathrm{and}  \\ 
C_{N,2} &= C_{N} \setminus C_{N,1}.
\end{split}
\end{align}

 Here, $C_N$ denotes the set of all candidate partition points minus a neighborhood of length $W/N$ on the ends. The set $C_{N,1}$ takes the union of the neighborhoods of all partition points, where  each neighborhood around a partition point is of radius $W/N$. The set $C_{N,2}$ denotes the set of all points located at a distance of least $W/N$ from all partition points. To establish large sample properties of the measure $\widehat{D}(\cdot)$, we assume frequency partition points are reasonably well-separated as follows.

\begin{assumption}[Separated partition points]
\label{as:separated}
There exists a constant $c \in (0,1/2)$ such that $\omega_1 > c$, $\omega_K < \frac{1}{2} - c$, $\underset{j=1,2,\hdots, K-1}{\textrm{min}} \; |\omega_j - \omega_{j+1}| \geq c$. 
\end{assumption}


With these assumptions, the following theorem describes the asymptotic behavior of the proposed estimator $\widehat{D}(\omega)$.

\begin{theorem}\label{thm:consistency_test_statistic}
Suppose that the conditions stated in Assumptions \ref{as:smooth}-\ref{as:separated} hold. Let $N \rightarrow \infty$ and $T \rightarrow \infty$ such that $N/T \rightarrow 0$ and $W/N \rightarrow c$. 
Then, 
\begin{itemize}
\item[(a)] for $\omega \in C_{N,2}$,
$
\widehat{D}(\omega) \overset{p}{\longrightarrow} 0,
$
\item[(b)] for $\omega \in P_K = \{ \omega_1 , \omega_2 , \hdots , \omega_K \}$,
$$
\widehat{D}(\omega) \overset{p}{\longrightarrow} \frac{1}{2 \pi c} \int_{0}^1 \int_{0}^{2 \pi c} \sum_{a,b=1}^p \Big( g_{a,b}(u,\omega - \lambda) - g_{a,b}(u,\omega+\lambda) \Big) \Big( g_{a,b}(u,\omega - \lambda) - g_{a,b}(u,\omega + \
\lambda) \Big)^{*} \; d \lambda \; du,
$$
 where $a,b=1,2,\hdots,p$, and $\overset{p}{\longrightarrow}$ denotes convergence in probability. 
\end{itemize}
 
\end{theorem}


 Proof is made available in the Appendix and relies on computing the asymptotic mean and variance of $\widehat{D}(\cdot)$. This result means that the discrepancy measure will approach zero for frequencies not near a partition point and approach a positive constant for frequencies representing partition points. 
The above result also implies that a hypothesis test using $\widehat{D}(\cdot)$ as the test statistic leads to a consistent test. Such a test would be defined in the following manner. Let $\eta_{T,\omega}$ be the threshold satisfying $\underset{T \rightarrow \infty}{\lim} \; \eta_{T,\omega} \geq B >0 $, for some positive constant $B$. Then, a candidate partition point $\omega \in (0, 1/2)$ is designated as a partition point if 
\begin{equation}\label{eq:test_threshold_description}
\widehat{D}(\omega) > \eta_{T,\omega}. 
\end{equation}  

 With the discrepancy measure $\widehat{D}(\omega)$ as the test statistic, the threshold $\eta_{T,\omega}$ can be determined by the bootstrap procedure described in Section \ref{sec:bootstrap_sections}. 

\subsection{Bootstrapping locally stationary processes}
\label{sec:bootstrap_sections}

Here we describe the resampling procedure to obtain the threshold $\eta_{T,\omega}$ given in \eqref{eq:test_threshold_description}. This threshold can be viewed as a critical value of the discrepancy measure $\widehat{D}(\omega)$ from \eqref{eq:main_discrepance_measure_estimated}, under the null hypothesis that $\omega$ is not a partition point. We resort to a nonparametric bootstrap method that generates samples under the null hypothesis in order to approximate $\eta_{T,\omega}$ and provide a corresponding p-value. 



Under the null hypothesis, assume $X_{t,T} = Y_t + \sigma(t/T) Z_t$ where $Y_t$ is a second-order stationary $p$-variate process, $\sigma(t/T)$ is a $p \times p$ time-varying matrix, $Z_t$ is i.i.d. $N(0,I_p)$, and $Y_t$ is independent of $Z_t$.  In this case, the spectral matrix of $X_{t,T}$ is $f_{X}(t/T,\omega) = f_Y(\omega)+\sigma(t/T)\sigma'(t/T)$, where $f_Y(\omega)$ is the spectral matrix of $Y_t$. However, $f_Y(\omega)$ does not appear in the demeaned time-varying spectral matrix \eqref{eq:demeaned_spectral_matrix} since $f_Y(\omega) = \int_0^1 f_Y(\omega)du$.  Therefore, we can simplify computation for our bootstrap procedure by avoiding estimation of $f_Y(\omega)$ and instead assuming $X_{t,T} = \sigma(t/T) Z_t$ in order to generate bootstrap samples.
Let $\omega \in C_N$ be a candidate frequency partition point.  The bootstrap procedure is carried out through the following steps. 

\begin{itemize}
    \item[Step 1.] Assume $X_{t,T} = \sigma(t/T) Z_t$, where $Z_t$ is i.i.d. $N(0,I_p)$. Compute the time-varying variance matrix estimator as 
    \begin{equation}\label{eq:tv-variance-matrix-estimator}
    \widehat{\Gamma}_{X,0} (u) = \frac{1}{T} \sum_{t=1}^T X_t X_t^{'} K_h(u - t/T),    
    \end{equation}
    where $K_h(u)  = \frac{1}{h}K(\frac{u}{h})$, $h$ denotes the bandwidth, and $K(\cdot)$ is a symmetric kernel function that integrates to 1. 
    
    \item[Step 2.] Compute the time-varying square root matrix $\widehat{\sigma}(u) = \Big( \widehat{\Gamma}_{X,0}(u) \Big)^{1/2} $. 
    
    \item[Step 3.] Obtain $R$ bootstrap resamples as $X_{t,T}^{(r)} = \widehat{\sigma}(t/T) Z_t^{(r)}$, where $Z_t^{(r)} \sim N(0,I_p)$, $t=1,2,\hdots,T$, and $r=1,2,\hdots R$.
    
    \item[Step 4.] With each resample $X_{t,T}^{(r)}$, $t=1,2,\hdots,T$, compute the discrepancy measure $\widehat{D}^{(r)}(\omega)$. 
    
    \item[Step 5.] Obtain the p-value of this test as $\frac{ \sum_{r=1}^R \mathbbm{1}_{ \widehat{D}^{(r)}(\omega) > \widehat{D}(\omega)} }{R}$, where $\widehat{D}(\omega)$ is the observed value of the test statistic. 
    
\end{itemize}

	





 Recall that the assumption in Step 1 follows from the behavior of the demeaned spectral matrix $g(u,\omega)$ under the null hypothesis that $\omega$ is not a partition point. In this case, $g(u,\omega) = g(u)$ which does not depend on frequency $\omega$. Under this assumption, Steps 2-5 then describe the resampling procedure that produces the required p-value. In Step 1, in finite sample situations discussed in Sections \ref{sec:simulations} and \ref{sec:application}, we utilize the triangular kernel for $K(\cdot)$ with a bandwidth $h = T^{-0.3}$.

\subsection{Finding the number and locations of partition points}
\label{sec:detecting_number_locations_cpts}

Here we describe the steps to detect the locations and number of  frequency partition points. Recall from \eqref{eq:points_sets_defn}, the set $C_{N}$ represents the set of candidate frequency partition points to search over. Our iterative procedure to detect the locations of partition points involves the following steps. 

\begin{itemize}

\item[Step 0.] Initialize the set $\widetilde{C}_N = C_N$ and $\widetilde{P} = \emptyset $, where $\widetilde{P}$ is the final set of partition points returned by the procedure. 

\item[Step 1.] Compute the measure $\widehat{D}(\omega)$, for every $\omega \in C_N$.  

\item[Step 2.] Find the point $\omega^{*}$, where
\begin{equation}
\omega^{*} = \textrm{arg.} \; \underset{\lambda \in \widetilde{C}_N}{ \textrm{max} } \; \widehat{D}(\lambda).
\end{equation}


\item[Step 3.] Determine the p-value for testing if $\omega^*$ is a partition point using the resampling procedure described in Section \ref{sec:bootstrap_sections}.  If found to be significant, set $\widetilde{P} = \widetilde{P} \bigcup
 \{ \omega^{*} \} $, and set $\widetilde{C}_N = \widetilde{C}_N \setminus  \{ \omega^{*} - \frac{W}{N}, \hdots , \omega^{*} + \frac{W}{N} \}$. 
 
 \item[Step 4.] Repeat Steps 2 and 3 until the significance test in Step 3 fails to return a significant frequency partition point.  
\end{itemize}
 The above iterative procedure results in the set $\widetilde{P}$ that contains the final set of frequency partition points, and the cardinality of this set provides an estimate $\widehat{K}$ of the number of partition points.  
Next, we provide a large sample result on the estimate $\widehat{K}$ of the number of partition points obtained through this procedure. 

\begin{theorem}\label{thm:consistency_number_of_cpts}
Suppose that the conditions stated in Assumptions \ref{as:smooth}-\ref{as:separated} hold. Let $N \rightarrow \infty$ and $T \rightarrow \infty$ such that $N/T \rightarrow 0$ and $W/N \rightarrow c$. 
Then, 

\begin{equation}
P( \widehat{K} \neq K  ) \overset{T \rightarrow \infty}{ \longrightarrow } 0.  
\end{equation} 
\end{theorem}

\begin{proof}
See Appendix for details of the proof.
\end{proof}

\subsection{Choice of W: length of neighborhood of frequencies}
\label{sec:choice_of_W}
The choice of $W$ used to estimate the discrepancy measure in \eqref{eq:main_discrepance_measure_estimated} depends on the nature and magnitude of changes in the frequency space.  Smaller changes need larger values of $W$ for detection while larger changes can be identified even with smaller values of $W$. In order to detect partition points associated with both small and large changes, we adopt a multiscale approach \citep{messer2014}. In practice, a sequence of $q$ choices for $W$ given by $W_{min} < W_1<W_2,...<W_q < W_{max}$ is considered. Let $\widetilde{P}_i$ denote the set of partition points estimated using the iterative procedure from Section \ref{sec:detecting_number_locations_cpts} with neighborhood length choice $W_i$. Set $P=\widetilde{P}_1$, where $P$ denotes the final set of estimated change points returned by our multiscale approach. For any point $\omega \in \widetilde{P}_2$, $\omega$ is added to the set $P$ only if it does not belong to a $W_2$-neighborhood of any of the existing points in the set $P$. The procedure is successively moved forward until all choices for $W$ have been considered. Algorithm \ref{algo:multipleW} provides the pseudocode that illustrates the full implementation of the multiscale frequency band estimation procedure. 

\begin{algorithm}[ht!]
	\DontPrintSemicolon
	
	$\widetilde{C}_N \gets C_N = \left\{ \frac{W_1}{N} , \frac{W_1+1}{N} , \hdots , \frac{1}{2} - \frac{W_1}{N}  \right\}$
	
	$\widetilde{P} \gets \emptyset$
 
		    
	
\For{$W \in \{W_1,W_2,\ldots,W_q\}$} 
{%
    
    $\widetilde{C}_N \gets \widetilde{C}_N \setminus \{\frac{W_1}{N}, \ldots,\frac{W}{N},\frac{1}{2} - \frac{W}{N},\ldots, \frac{1}{2} - \frac{W_1}{N}\}$
    
    \lFor{$\lambda \in \widetilde{P}$}{$\widetilde{C}_N \gets \widetilde{C}_N \setminus  \{ \lambda - \frac{W}{N}, \hdots , \lambda + \frac{W}{N} \}$}
    
 $\mathrm{stop} \gets 0$
   
    \While{$\mathrm{stop}=0 \; \mathrm{and} \; \widetilde{C}_N \ne \emptyset$}{
 
        $\omega^{*} \gets \argmax_{\lambda \in \widetilde{C}_N}  \widehat{D}(\lambda) \; \mathrm{where} \; \widehat{D}(\lambda) \; \mathrm{is} \; \mathrm{calculated} \; \mathrm{by} \; \eqref{eq:main_discrepance_measure_estimated} \; \mathrm{given} \; W$

    Determine $p$-value for $\widehat{D}(\omega^*)$ using bootstrap procedure (see Section \ref{sec:bootstrap_sections})

 \uIf{$p\mathrm{-value} \; \mathrm{is} \; \mathrm{significant}$}{

            $\widetilde{P} \gets \widetilde{P} \bigcup \{ \omega^{*} \} $ 
            
            $\widetilde{C}_N \gets \widetilde{C}_N \setminus  \{ \omega^{*} - \frac{W}{N}, \hdots , \omega^{*} + \frac{W}{N} \}$}
        \lElse{$\mathrm{stop} \gets 1$}
    }
}
$\widehat{K} = |\widetilde{P}|$ 

\Return{$\widetilde{P}, \widehat{K}$}
\caption{{\sc Multiscale Frequency Band Estimation}}
\label{algo:multipleW}    
\end{algorithm}

 It should be noted that $W_{min}$ should be selected small enough to ensure partition points associated with more subtle changes are detected, but not too small, which may lead to false positives.  
In finite sample cases in Section \ref{sec:simulations}, we present the performance results for different choices of $W_{min}$.  Based on our simulation results, $W_{min} = N/8$ provides the best estimation performance for the simulation settings considered herein.


 In scenarios where a single frequency neighborhood length choice $W$ must be selected, one can take the largest value of $W \in \{ W_1,W_2,...,W_q \}$ for which there is an addition of a partition point to the set $P$ in the iterative procedure described above. More precisely $W = W_{i^*}$, where $i^{*}$ is the largest value in the set $\{ 1,2,...,q \}$ for which the iterative procedure described above adds a point to the set $P$ during iteration $i^*$ (i.e., with neighborhood length choice $W_{i^{*}}$). In case there is no point added to the set $P$ for any choice $W_i$, $i= 1,2,...,q$, we set $W=W_q$. 

\subsection{Finding components responsible for partition points} 
\label{sec:components_partition_point}
 In this section, we present a new technique to identify the components and cross-components of the multivariate series that significantly contribute to each of the partition points identified in the frequency space. For every identified partition point, a resampling procedure for finding the components significantly contributing to the change characterized by the frequency partition point is discussed. The proposed approach, similar to the bootstrap method given in Section \ref{sec:bootstrap_sections}, generates samples under the null hypothesis, and results in p-values for every component and cross-component of the  series $X_{t,T}$.  

 Let $\omega_c$ be a partition point detected by our method. With every component $(a,b)$, where $1\leq a \leq b \leq p $, the goal is to estimate the p-value corresponding to the null hypothesis that component $(a,b)$ of the series $X_{t,T}$ does not have a significant contribution to the partition at frequency $\omega_c$. The component-specific test statistic $\widehat{D}_{(a,b)}$ is then written as 
\begin{equation}\label{eq:component_specific_discrepance_measure_estimated}
\widehat{D}_{(a,b)}(\omega_c) = \frac{1}{T} \sum_{t=1}^T \frac{1}{W} \sum_{k=1}^{W} \left| \widehat{g}_{a,b} \Big( t/T , \omega_c - \lambda_k \Big) - \widehat{g}_{a,b} \Big( t/T , \omega_c + \lambda_k \Big) \right|^2, 
\end{equation}  
where $\widehat{g}( t/T , \omega_c - \lambda_k ) = I_N( t/T , \omega_c - \lambda_k ) - \frac{1}{T} \sum_{t_1=1}^{T} I_N( \frac{t_1}{T} , \omega_c - \lambda_k ) $ and $\lambda_k = \frac{k}{N}, \; k=1,2,\hdots,W$. Note that $\widehat{g}_{a,b}$ is component $(a,b)$ of the $p \times p$ matrix $\widehat{g}(\cdot)$.  The p-value is obtained through the following steps. 

\begin{itemize}
    \item[Step 1.] Assume $X_{t,T} = \sigma(t/T) Z_t$, where $Z_t$ is i.i.d. $N(0,I_p)$. Compute the time-varying variance matrix estimator as 
    \begin{equation}\label{eq:tv-variance-matrix-estimator}
    \widehat{\Gamma}_{X,0} (u) = \frac{1}{T} \sum_{t=1}^T X_t X_t^{'} K_h(u - t/T),    
    \end{equation}
    where $K_h(u)  = \frac{1}{h}K(\frac{u}{h})$, $h$ denotes the bandwidth, and $K(\cdot)$ is a symmetric kernel function which integrates to 1. 
    
    \item[Step 2.] Compute the time-varying square root matrix $\widehat{\sigma}(u) = \Big( \widehat{\Gamma}_{X,0}(u) \Big)^{1/2} $. 
    
    \item[Step 3.] Obtain $R$ bootstrap resamples as $X_{t,T}^{(r)} = \widehat{\sigma}(t/T) Z_t^{(r)}$, where $Z_t^{(r)} \sim N(0,I_p)$, $t=1,2,\hdots,T$, and $r=1,2,\hdots R$.
    
    \item[Step 4.] With each resample $X_{t,T}^{(r)}$, $t=1,2,\hdots,T$, compute the component-specific test statistic $\widehat{D}^{(r)}_{(a,b)}(\omega_c)$. 
    
    \item[Step 5.] Obtain the p-value of this test as $\frac{ \sum_{r=1}^R \mathbbm{1}_{ \widehat{D}_{(a,b)}^{(r)}(\omega_c) > \widehat{D}_{(a,b)}(\omega_c)} }{R}$, where $\widehat{D}_{(a,b)}(\omega)$ is the observed value of the test statistic. 
    
\end{itemize}

 The above procedure is applied to every component and cross-component $(a,b)$, $1 \leq a \leq b \leq p$, of the series $X_{t,T}$. The components and cross-components that carry a significant p-value are deemed as the components responsible for the partition at frequency $\omega_c$. An illustration of this procedure can be seen in the application presented in Section \ref{sec:application}. In order to account for simultaneous testing of multiple components, multiple testing adjustments can and should be used to control the experiment-wide error rate;  for example, a Bonferroni adjustment is implemented for the application presented in Section \ref{sec:application}.

\section{Simulation study}
\label{sec:simulations}

 Performance of the proposed method is assessed through a few simulation examples. The five simulation schemes are described first followed by the presentation of the performance results.

 The first scheme (WN1B) is a multivariate white noise model with no partition points in the frequency space. This setting is considered to ensure that the method does not produce an unreasonable number of false positives.  The second scheme (L3B) considers partition points at frequencies 0.15 and 0.35, with spectral density $f_2(u,\omega)$ exhibiting a linear trend in time $u \in (0,1)$. The third scheme (S3B) also considers partition points at frequencies 0.15 and 0.35, but with spectral density $f_3(u,\omega)$ exhibiting a non-linear trend in time.  The second and third schemes illustrate performance of the method in capturing partition points associated with both linear and nonlinear time-varying dynamics.
The fourth scheme (M3B-1), with partition points at frequencies 0.15 and 0.35, considers a mixture of time series exhibiting linear and non-linear trends from the models L3B and S3B, respectively. This setting considers performance of the method when components of the series have a similar frequency band structure, but differing time-varying dynamics across components.
The fifth scheme (M3B-2) again considers a  mixture of linear and non-linear trends in the spectral density, but assumes only $20\%$ of the $p$ components in $X_{t,T}$ contribute to the partition at frequency 0.15, whereas the remaining $80\%$ of the $p$ components contribute to the partition at frequency 0.35.  This setting is the most challenging and represents both differing time-varying dynamics and differing partition points across components.

   \begin{enumerate}
    \setlength\itemsep{0.5pt}

\item {\bf White noise  (WN1B). } $X_{k,t} = z_{1,t+k-1}$ for $k=1,2,\ldots,p$, and $z_{1,t}$ has time-varying spectral density $f_1(u,\omega)$ given by

\begin{equation*}
f_1(u,\omega) = 1 \; \mathrm{for} \; \omega \in
(0,0.5)
\end{equation*}
        
\item{\bf Linear, 3-Bands (L3B).} $X_{k,t} = z_{2,t+k-1}$ for $k=1,2,\ldots,p$, and $z_{2,t}$ has time-varying spectral density $f_2(u,\omega)$ given by

\begin{equation*}
f_2(u,\omega) = 
\begin{cases}
10-9u \; \mathrm{for}  \; \omega \in (0,0.15)\\
1 \; \mathrm{for}  \; \omega \in [0.15,0.35)\\
1+9u \; \mathrm{for} \; \omega \in [0.35, 0.5)\\
\end{cases}
\end{equation*}

\item{\bf Sinusoidal, 3-Bands (S3B).} $X_{k,t} = z_{3,t+k-1}$ for $k=1,2,\ldots,p$, and $z_{3,t}$ has time-varying spectral density $f_3(u,\omega)$ given by 

\begin{equation*}
f_3(u,\omega) = 
\begin{cases}
10+10\sin(4 \pi u - \pi/2) \; \mathrm{for}  \; \omega \in (0,0.15]\\
5+5\cos(4 \pi u) \; \mathrm{for}  \; \omega \in (0.15,0.35]\\
8.5+8.5\sin(3 \pi u-\pi/16) \; \mathrm{for} \; \omega \in (0.35, 0.5)\\
\end{cases}
\end{equation*}
        
\item{\bf Linear and Sinusoidal, 3-Bands, Mixture (M3B-1).} $X_{k,t} = z_{2,t+k-1}$ for $k=1,2,\ldots,\lfloor p/2 \rfloor$, and $X_{k,t} = z_{3,t+k-\lfloor p/2 \rfloor-1}$ for $k=\lfloor p/2 \rfloor+1,\ldots,p$. Here, the series $z_{2,t}$ and $z_{3,t}$ are given by the schemes L3B and S3B described above.     
        
\item{\bf Linear and Sinusoidal, 3-Bands, Differing Proportions (M3B-2). } $X_{k,t} = z_{4,t+k-1}$ for $k=1,2,\ldots,\floor{0.2p}$, and $X_{k,t} = z_{5,t}$ for $k=\floor{0.2p}+1 \ldots , p$. Here, $z_{4,t}$ and $z_{5,t}$ have  time-varying spectral densities $f_4(u,\omega)$ and $f_5(u,\omega)$, respectively. 

\begin{equation*}
f_4(u,\omega) = 
\begin{cases}
10-9u \; \mathrm{for}  \; \omega \in (0,0.15)\\
1 \; \mathrm{for}  \; \omega \in [0.15,0.5)\\
\end{cases}
\end{equation*}

\begin{equation*}
f_5(u,\omega) = 
\begin{cases}
5+5\cos(4 \pi u) \; \mathrm{for}  \; \omega \in (0,0.35]\\
8.5+8.5\sin(3 \pi u-\pi/16) \; \mathrm{for} \; \omega \in (0.35, 0.5)\\
\end{cases}
\end{equation*}
        
\end{enumerate}

 To assess the performance of the proposed method, we first present the results on estimating the true number of frequency bands, i.e., the quantity $K+1$, where $K$ is the true number of partition points defined in \eqref{eq:g_frequency_partition}. Table \ref{tab:numberbands1} reports the estimated mean number of frequency bands for the five simulation schemes based on 100 replications. Note that the bootstrap procedure from Section \ref{sec:detecting_number_locations_cpts} is used for estimating the number and locations of the partition points. In implementing this procedure, the triangular kernel is used as the kernel choice $K(\cdot)$, and the bandwidth is $h=T^{-0.3}$. The choice for the frequency neighborhood length $W$ involves the multiscale approach described in Section \ref{sec:choice_of_W}, and we consider an equally-spaced sequence $W_{min} = \frac{N}{8} < W_1<W_2,...<W_q < W_{max} = \frac{N}{4}$, with $N=T^{0.7}$, and the sample size $T \in \{ 200,500,1000 \}$. The results in Table \ref{tab:numberbands1} show that as the sample size increases, accuracy in estimating the number of partition points increases for all five simulation schemes. Table \ref{tab:numberbands2} presents estimation results at a fixed sample size $T=1000$, but for different choices of $W_{min}$ used in the multiscale procedure from Section \ref{sec:choice_of_W}. It is seen that in most schemes the number of frequency bands estimated increases as $W_{min}$ decreases, with the best results at parameter choice $W_{min} = N/8$.   

\begin{table}[ht!]
    \centering
   \begin{tabular}{  c c| c|c| c|c |c }
\hline 
$p$ & \multicolumn{1}{c|}{T} & WN1B & L3B & S3B & M3B-1 & M3B-2 \\
\hline
 \multirow{3}{*}{10} 
 &  200  & 1(0) & 2.24(0.55) & 2.15(0.36) & 2.09(0.35)& 2.14(0.64) \\
 & 500   & 1(0) & 2.94(0.34) & 2.56(0.50) & 2.47(0.52)&2.88(0.57) \\
 & 1000 & 1(0) & 3.06(0.24) & 2.96(0.24) & 2.92(0.31)& 3.09(0.35)\\
 \hline
\multirow{3}{*}{15} 
&   200  & 1(0) & 2.34(0.52) & 2.18(0.39) & 2.07(0.29)& 2.25(0.59)\\
 & 500  & 1(0) & 2.99(0.30) & 2.61(0.55) & 2.53(0.52)&2.92(0.51) \\
 & 1000  & 1(0) & 3.05(0.22) & 3.00(0.20) & 2.98(0.20)& 3.12(0.41)\\
\end{tabular}
    \caption{Mean(sd) for estimated number of frequency bands, $\hat{K}+1$, for 100 replications ($W_{min} = N/8$). True value $K+1$ is 1 for WN1B, and 3 for all other schemes.  }
    \label{tab:numberbands1}
\end{table}

\begin{table}[ht!]
    \centering
   \begin{tabular}{  c c| c|c| c|c |c }
\hline 
$p$ & \multicolumn{1}{c|}{$W_{min}$} & WN1B & L3B & S3B & M3B-1 & M3B-2 \\
\hline
 \multirow{3}{*}{10} 
 & $N/8$  & 1(0) & 3.06(0.24) & 2.96(0.24) & 2.92(0.31)& 3.09(0.35) \\
 & $N/10$ & 1(0) & 3.64(0.64) & 3.43(0.52)  & 3.13(0.34)& 3.77(0.47)\\
 & $N/12$ & 1.07(0.26) & 3.36(0.58) & 4.52(0.56) & 4.11(0.71)& 4.26(0.66)\\
 \hline
\multirow{3}{*}{15} 
 & $N/8$ & 1(0) & 3.05(0.22) & 3.00(0.20) & 2.98(0.20) & 3.12(0.41)\\
 & $N/10$ & 1(0) & 3.69(0.66) & 3.50(0.54) & 3.22(0.42)& 3.87(0.42)\\
 & $N/12$ & 1.06(0.24)  & 3.44(0.59) &  4.68(0.51) & 4.32(0.63) & 4.41(0.64) \\
\end{tabular}
    \caption{Mean(sd) for estimated number of frequency bands, $\hat{K}+1$, for 100 replications ($T=1000$). True value $K+1$ is 1 for WN1B, and 3 for all other schemes.}
    \label{tab:numberbands2}
\end{table}

 Next, we present the proportion of the 100 replications that result in \textit{correct detection}. At any given replication, a \textit{correct detection}  occurs when the proposed method identifies the  correct number of frequency partition points \emph{and} all estimated partition points are within a distance of $\zeta$ from the true partition point. Table \ref{tab:accuracy} presents the \textit{correct detection} rate based on 100 replications for the four simulation schemes with more than one frequency band. We observe that in almost all cases, as the length of the time series increases, the \textit{correct detection} rate increases. 

The scheme M3B-2 exhibits the lowest \textit{correct detection} rate relative to other settings. Recall that under the M3B-2 setting, only 20\% of the $p$ components of the multivariate series $X_{t,T}$ contribute to the partition at frequency 0.15, and the remaining 80\% of the components contribute to the partition at frequency 0.35. With weaker contribution from components towards the frequency partition point 0.15, the proposed method requires much longer time series ($T$) to improve the \textit{correct detection} rate.  
\begin{table}[ht!]
    \centering
    \begin{tabular}{  c c| c|c| c|c  }
\hline 
$p$ & \multicolumn{1}{c|}{T} &  L3B & S3B & M3B-1 & M3B-2 \\
\hline
 \multirow{3}{*}{10} 
 &  200  & 0.3 & 0.15 & 0.11 & 0.17\\
 & 500   & 0.88 & 0.53 & 0.45 & 0.36\\
 & 1000 &0.94 & 0.93 & 0.90 & 0.42\\
 \hline
\multirow{3}{*}{15} 
&   200  &  0.36 & 0.18 & 0.08 & 0.20\\
 & 500  &  0.91 & 0.52 & 0.51 & 0.41\\
 & 1000  &  0.95 & 0.96 & 0.96 & 0.39\\
\end{tabular}
    \caption{\textit{Correct Detection}. Proportion of 100 replications that correctly estimate the number of   bands, $K+1$, and the distance between the estimated and true partition points is no more than $\zeta$. Here $\zeta = 1/16$ and $W_{min}=N/8$.}
    \label{tab:accuracy}
\end{table}

\begin{table}[ht!]
    \centering
    \begin{tabular}{  c c| c|c| c|c  }
\hline 
$p$ & \multicolumn{1}{c|}{$\zeta$} &  L3B & S3B & M3B-1 & M3B-2 \\
\hline
 \multirow{3}{*}{10} 
 & 1/12 &0.94 &0.93 &0.90 & 0.73\\
 & 1/16 &0.94 &0.93 & 0.90& 0.42\\
 & 1/24 &0.94 &0.93 & 0.90& 0.27\\
 \hline
\multirow{3}{*}{15} 
 & 1/12 &0.95 & 0.96 &0.96 & 0.67\\
 & 1/16 &0.95 &0.96 &0.96 & 0.39\\
 & 1/24 &0.95 &0.96 &0.96 &0.25 \\
\end{tabular}
    \caption{\textit{Correct Detection}. Proportion of 100 replications that correctly estimate the number of   bands, $K+1$, and the distance between the estimated and true partition points is no more than $\zeta$. Here $T=1000$ and $W_{min} = N/8$.}
    \label{tab:accuracy2}
\end{table}

Finally, in Table \ref{tab:accuracy2}, we provide the results on \textit{correct detection} rates for different choices of $\zeta$ and fixed time series length $T=1000$ for the four settings with more than one frequency band. We notice that for the first three simulation schemes, the \textit{correct detection} rate is not sensitive to the choice of $\zeta$. Here again, the scheme M3B-2 sees the lower \textit{correct detection} rates. This can be attributed to the fact that fewer components (only 20\%) are contributing to the partition at frequency 0.15, and hence a much larger sample size is needed to achieve higher \textit{correct detection} rates. 

\section{Application}
\label{sec:application}

To illustrate the usefulness of the proposed method in analyzing multivariate biomedical time series, we turn to frequency band analysis of EEG signals.  Frequency bands are commonly used in the scientific literature to generate summary measures of the EEG signal, so a principled approach to frequency band estimation would be a welcomed development.  

Our analysis considers 16-channel EEG signals from 14 participants in a simple resting-state eyes-open, eyes-closed experimental protocol \citep{gregoire_cattan_2018_2348892}.  For computational efficiency, signals are standardized and downsampled to a sampling rate of 64 Hz and subset to include 5 consecutive blocks alternating between eyes-closed and eyes-open conditions lasting approximately 15 seconds per block. This  produces, for each participant, a time series approximately 4800 observations in length.  For illustration, Figure \ref{fig:intro_figure} displays  standardized EEG time series from the Oz occipital channel in participants 2 and 13.  The time-varying behavior can be witnessed around the 10 Hz frequency, which corresponds to the traditional Alpha frequency band (8-12 Hz) and is associated with the alternating eyes-closed and eyes-open conditions.  

In practice, the detection of Alpha waves is a useful indicator 
 of stress levels, concentration, relaxation, or mental load \citep{banquet1973spectral,antonenko2010using}.  However, Alpha power may vary in both its peak frequency and range of frequencies
across participants \citep{Doppelmayretal1998,eeg-peak-alpha-freq-changes}, so a data-driven approach is essential for accurately characterizing Alpha power across different individuals.  To illustrate this, Figure \ref{fig:partitionpoints} displays the estimated frequency partition points using the proposed method applied on each participant's EEG signal  separately.  A triangular kernel is used as the kernel choice $K(\cdot)$ with bandwidth $h=T^{-0.3}$, and the multiscale approach described in Section \ref{sec:choice_of_W} was used considering a sequence of equally-spaced values $W_{min} = \frac{N}{12} < W_1 < W_2 < W_3 < W_4 < W_5 < W_{max} = \frac{N}{4}$, with $N$ = $T^{0.7}$.

\begin{figure}[ht!]
    \centering
    \includegraphics[width=0.75\textwidth]{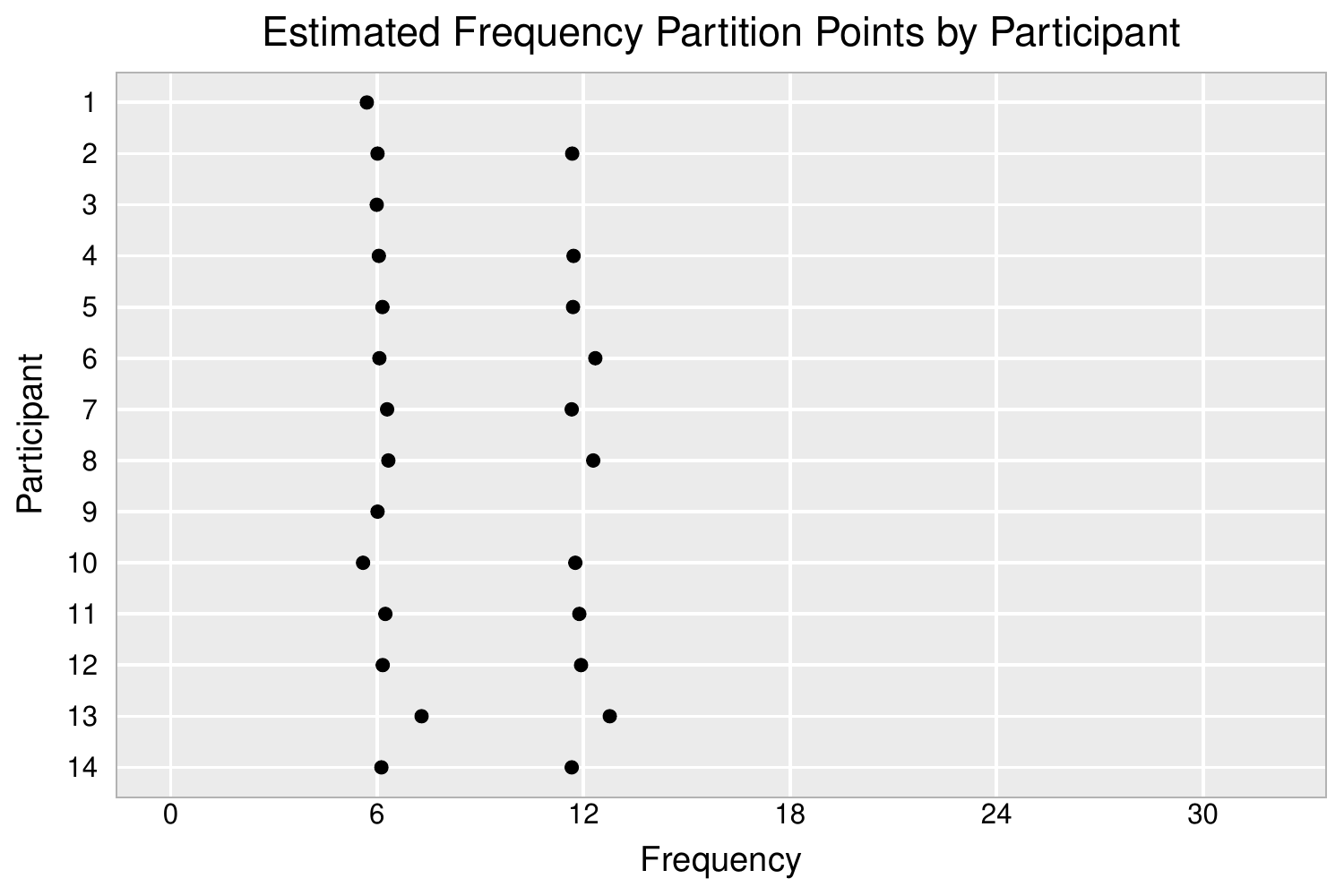}
    \caption{Frequency partition points estimated using proposed methodology.}
    \label{fig:partitionpoints}
\end{figure}

While frequency bands identified are similar across participants, some participants (e.g. Participant 13) exhibit Alpha power in a slightly higher band of frequencies than the conventional 8-12 Hz band used in practice, which demonstrates the advantage of the proposed data-driven frequency band estimator.  For further illustration, Figure \ref{fig:periodogram_p13} displays the estimated spectral density for Participant 13 in two channels, one from the occipital region (Oz) and another from the parietal region (P7), along with the estimated frequency partition points (green lines).   It is not surprising to find that the time-varying behavior in this band is prominent for these two channels, since the parietal and occipital brain regions are known to exhibit strong Alpha band power \citep{pfurtscheller1996event}.

\begin{figure}
    \centering
   \begin{subfigure}[b]{0.48\textwidth}
       \includegraphics[width=\textwidth,trim={0.5cm 0.6cm 1cm 0.75cm},clip]{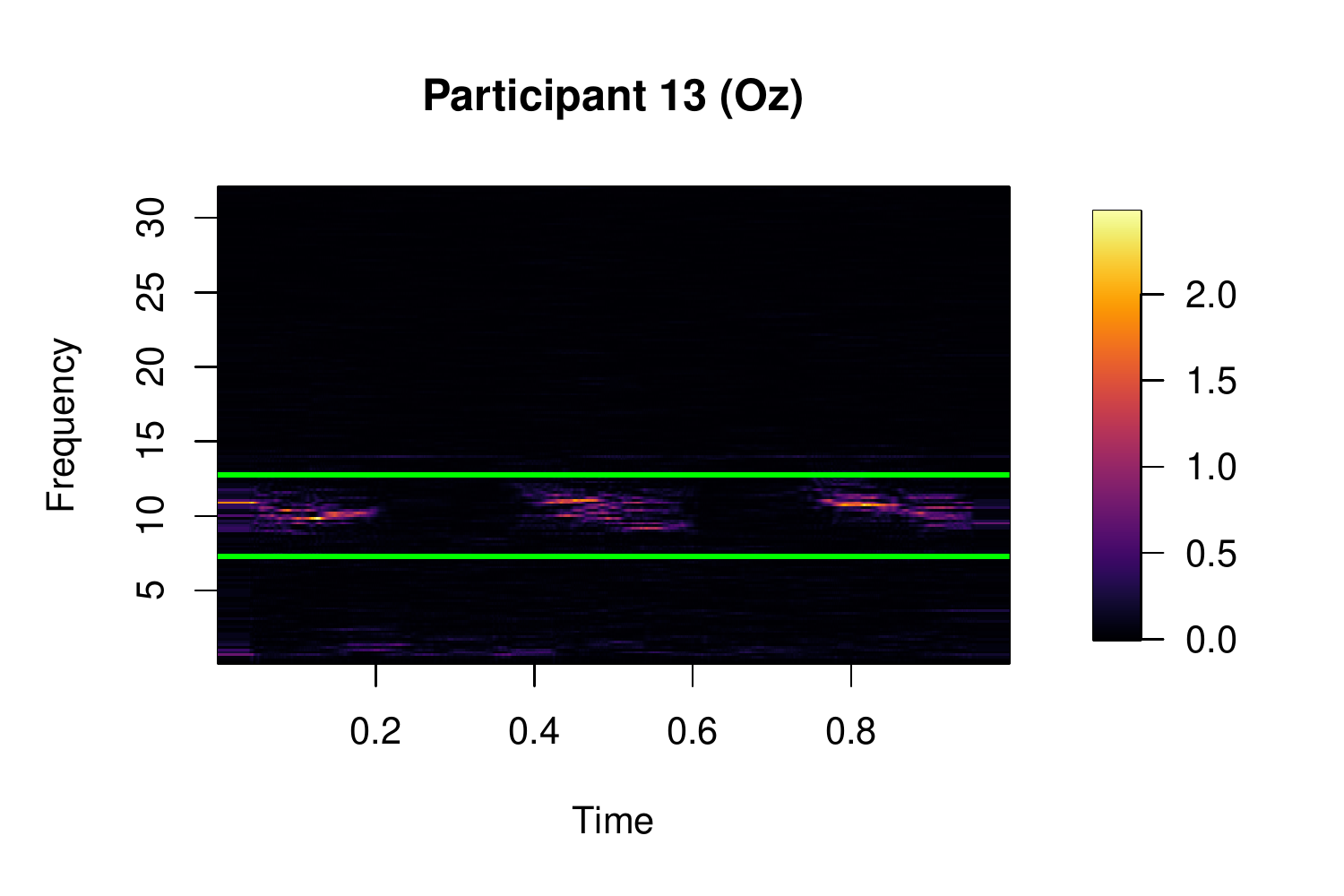}
    \caption{}
    \label{fig:periodogram_p13_oz}
   \end{subfigure}
   \begin{subfigure}[b]{0.48\textwidth}
          \includegraphics[width=\textwidth,trim={0.5cm 0.6cm 1cm 0.75cm},clip]{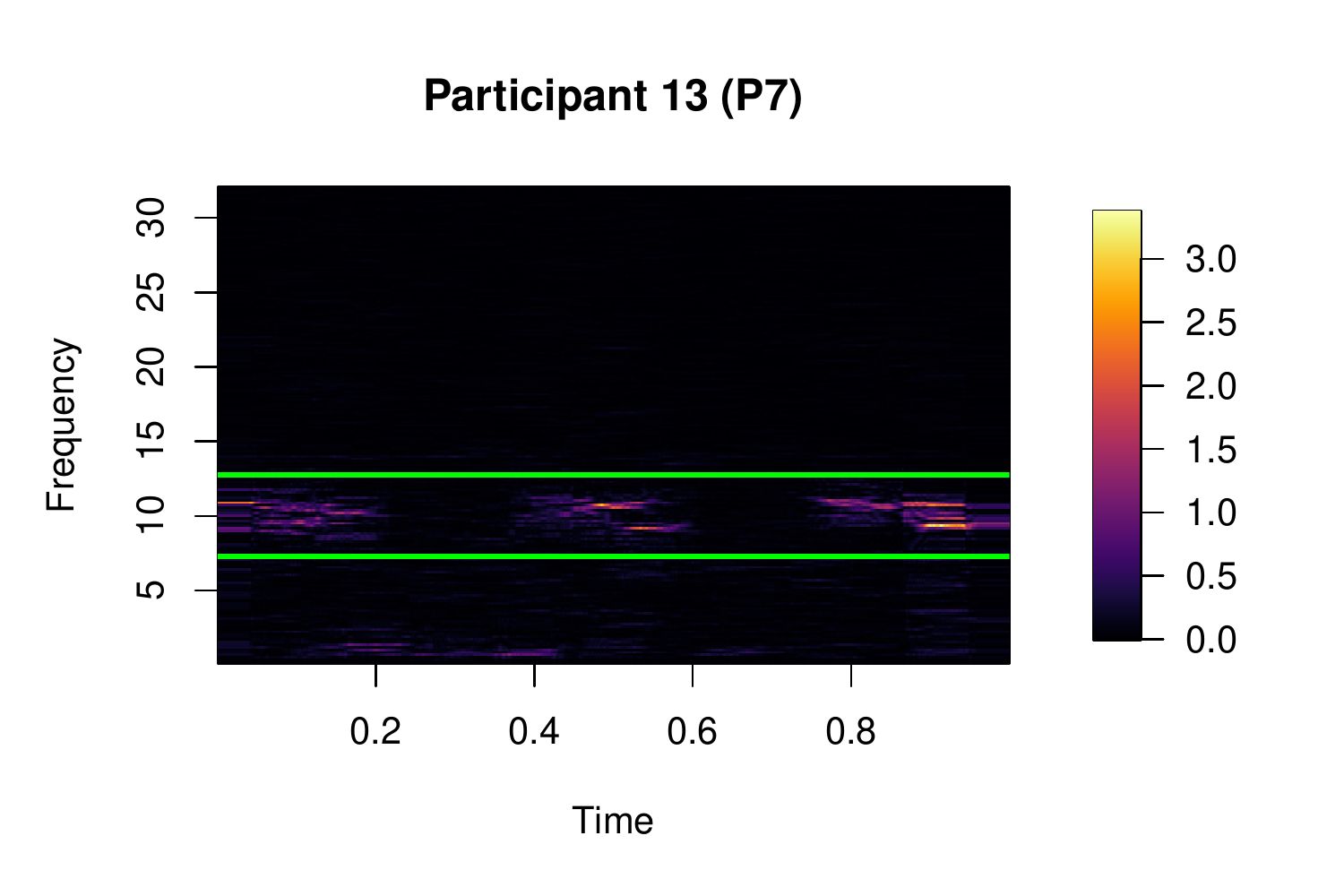}
              \caption{}
    \label{fig:periodogram_p13_p7}
   \end{subfigure}
    \caption{Local periodogram for two components and estimated frequency bands.}
    \label{fig:periodogram_p13}
\end{figure}

Next, in addition to identifying the frequency bands, the proposed method also provides a data-driven approach to identifying components and cross-components of the multivariate signal that are significantly associated with each identified  frequency partition point.  Applying the bootstrap technique described in Section \ref{sec:components_partition_point}, we can identify components that significantly contribute to the upper and lower Alpha power frequency partition points identified by the proposed method for Participants 2 and 13 (see Figure \ref{fig:contributions_p2p13}).  

\begin{figure}
    \centering
    \includegraphics[width=0.95\textwidth]{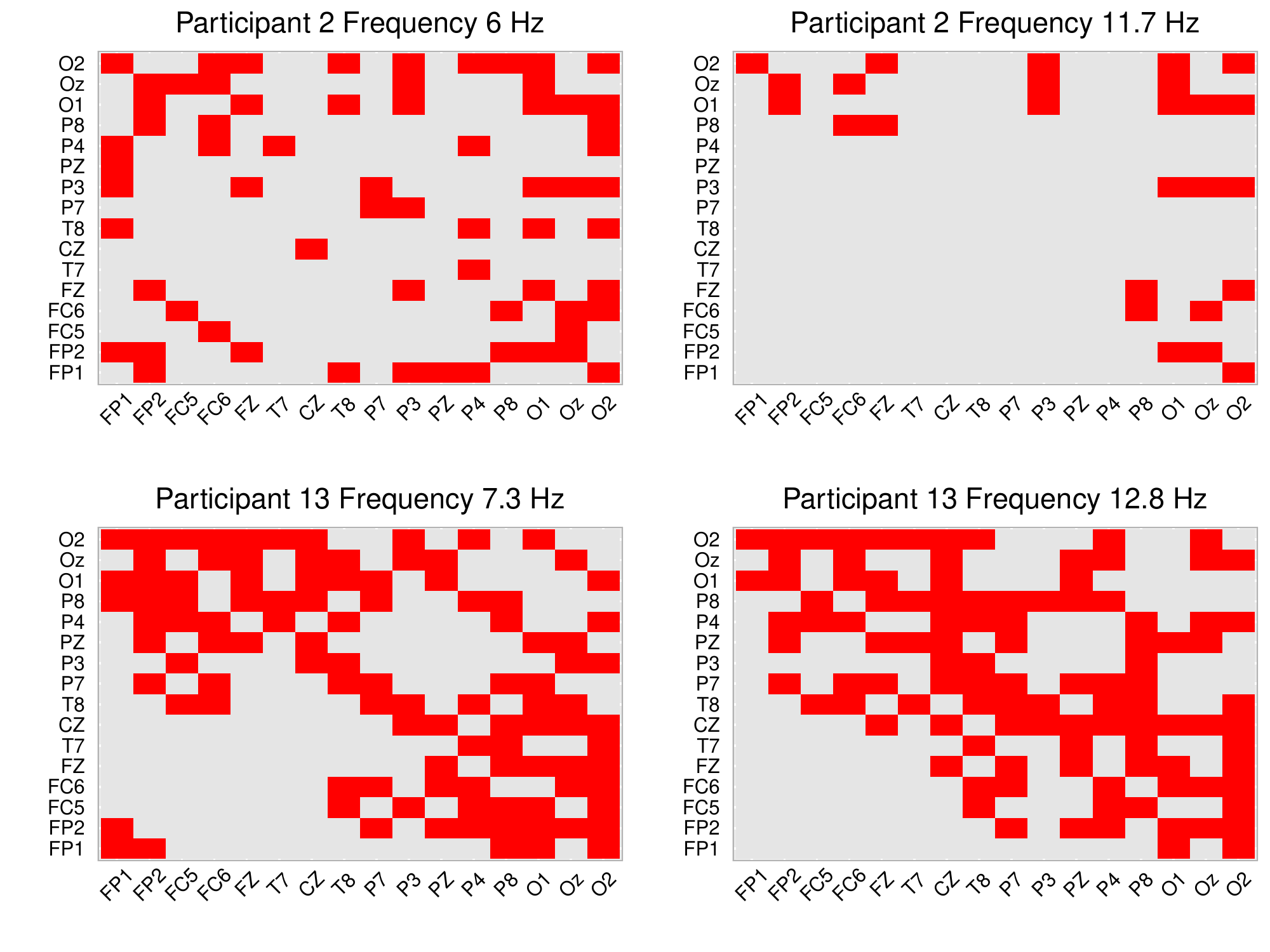}
    \caption{Components significantly contributing to estimated frequency partition point (red) for two participants and two partition points.}
    \label{fig:contributions_p2p13}
\end{figure}

Unsurprisingly, we find that channels from the parietal (P7, P3, Pz, P4, P8) and occipital (O1, Oz, O2) regions play a major role in establishing both the upper and lower bounds of the Alpha band.  However, there are two additional findings of interest.  First, in differentiating the time-varying dynamics of lower frequency power vs. Alpha band power (i.e., components responsible for frequency partition points at 6 Hz and 7.3 Hz for Participants 2 and 13 respectively), the prefrontal channels (FP1, FP2) are also involved.  The prefrontal region has been shown to be the dominant source of power in lower frequencies ($<$ 4 Hz) during both eyes-open and eyes-closed resting state conditions. Further, low frequency power has been shown to increase from the eyes-closed to eyes-open condition \citep{chen2008eeg}.  Since the prefrontal region is not meaningfully involved in Alpha power fluctuations, 
the differences in the time-varying dynamics for lower frequencies vs. Alpha power frequencies in the prefrontal region are shown to be significant using the proposed method. More precisely, our method shows that the prefrontal region contributes towards the identification of the lower Alpha band frequency partition points (6 Hz and 7.3 Hz for Participants 2 and 13 respectively), but not the upper Alpha band frequency partition points (11.7 Hz and 12.8 Hz for Participants 2 and 13
respectively).  Second, the cross-components between the parietal and occipital region channels and the prefrontal, frontal, and central region channels contribute differently to the lower and upper Alpha band frequency partition points.  This suggests that interaction components between the anterior and posterior brain regions also exhibit different time varying behavior in the lower and higher frequencies surrounding the Alpha band. 

Across all participants, Figure \ref{fig:contributions_secondcutpoint} displays the average component-wise p-values for testing the significance of the contribution towards the lower and upper Alpha band frequency partition points that were identified.  Here we can see more clearly that the lower frequency partition point, which separates the lower frequencies from the Alpha band, can be attributed to the prefrontal components and the interaction components between the prefrontal and parietal/occipital regions.  However, the upper frequency partition point separating the Alpha band from higher frequencies is associated with the parietal and occipital components and cross-components.  These findings are made possible by the proposed method that is uniquely able to identify frequency partition points and also the corresponding sets of significant components and cross-components associated with each partition point. 

\begin{figure}
    \centering
    \begin{subfigure}[b]{0.45\textwidth}
    \includegraphics[width=\textwidth]{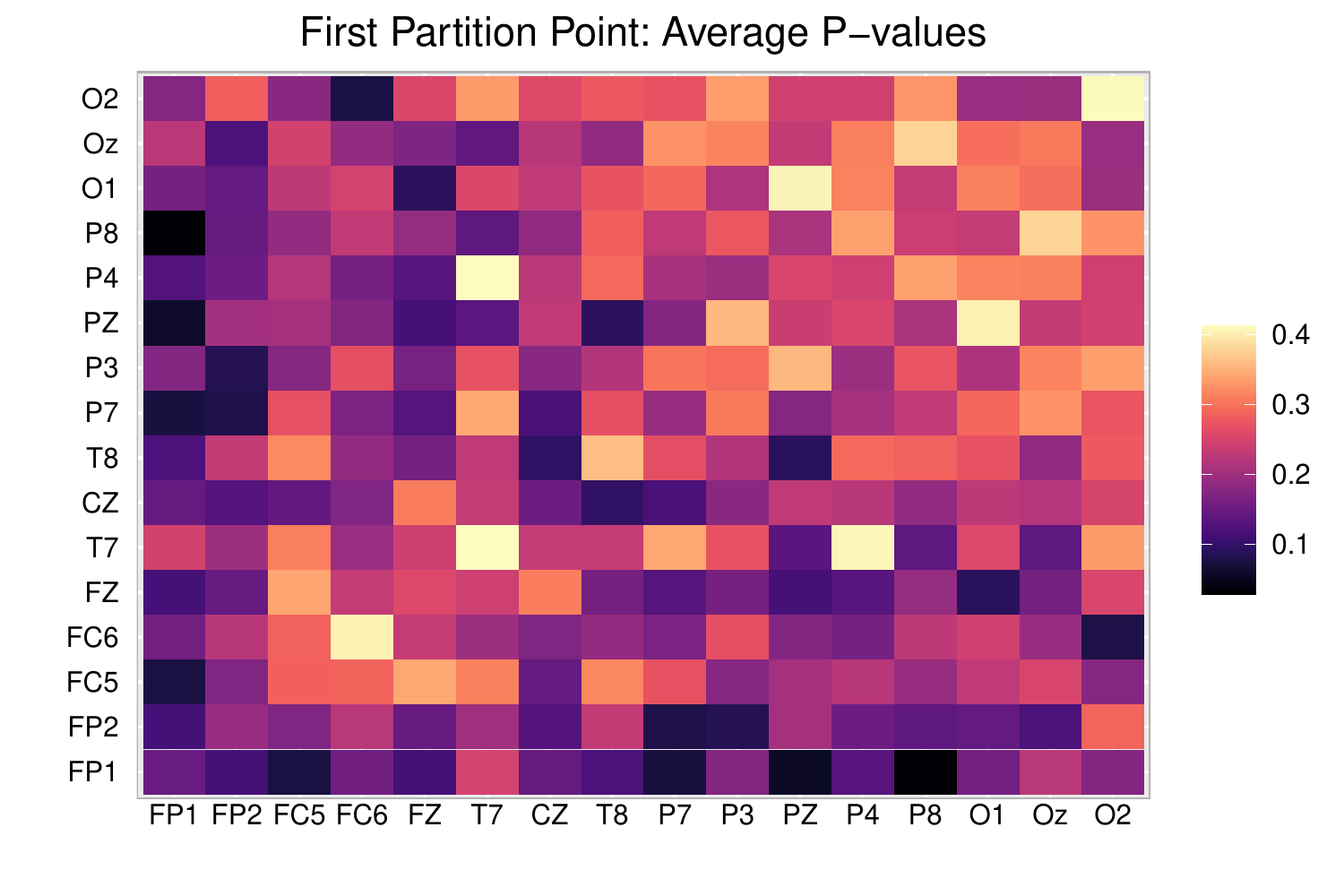}
         \caption{}
    \label{fig:contributions_firstcutpoint}
    \end{subfigure}
        \begin{subfigure}[b]{0.45\textwidth}
    \includegraphics[width=\textwidth]{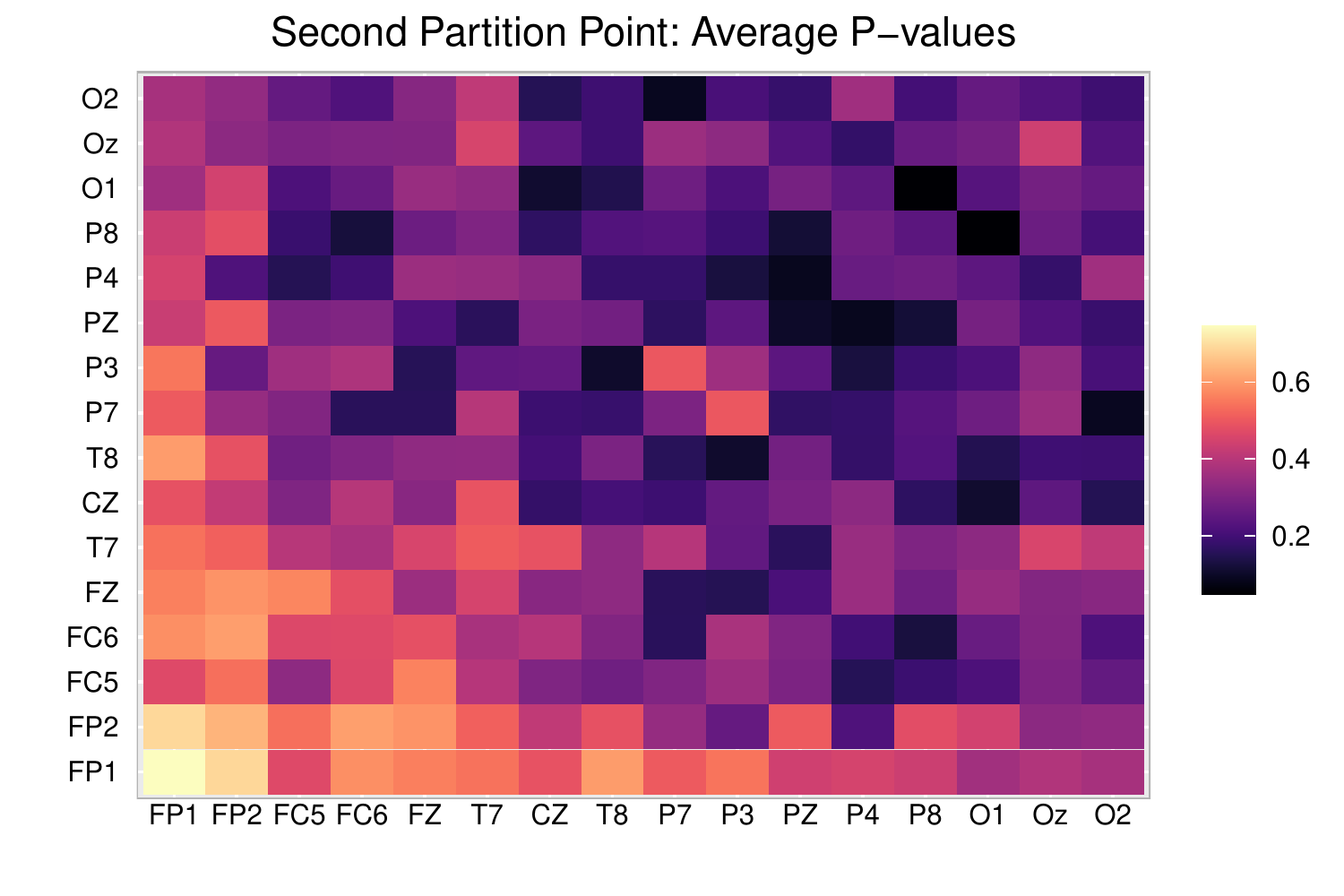}
         \caption{}
  \label{fig:contributions_secondcutpoint}
    \end{subfigure}
         \caption{Component-wise significance of contribution towards partition points}
  \label{fig:contributions_secondcutpoint}
\end{figure}

\section{Concluding remarks}
\label{s:conclusion}

The frequency band analysis framework introduced in this article offers a quantitative approach to identifying frequency bands that best preserve the nonstationary dynamics of the underlying multivariate time series.  This framework allows for estimation of both the number of frequency bands, their corresponding frequency partition points, and the components and cross-components of the multivariate signal associated with each of the partition points.  This is made possible by the development of a sensible discrepancy measure and computationally efficient bootstrap testing procedure within an iterative search algorithm.  However, the proposed method is not without limitations. 
Motivated by the application to EEG frequency band analysis, it would be interesting to extend this framework to directly consider multiple subjects and produce a single set of frequency bands that jointly characterizes the time-varying dynamics of the collection of signals in a data-driven manner.  Second, in order to extend the proposed methodology for analyzing high-dimensional EEG signals (64 to 512 channels), the bootstrap testing procedure would need to be modified to accommodate high-dimensional covariance structures.  This would require appropriate simplifying assumptions on the covariance structure, such as factor-based, sparse, and block covariance structures.  Finally, as seen in the last simulation setting (M3B-2) in Section \ref{sec:simulations}, the current discrepancy measure requires large amounts of data to detect frequency partition points associated with changes in a few components of the multivariate signal.  Modifications of the discrepancy measure that may be more powerful in detecting such changes (e.g. weighted squared $L_2$ norm or $L_{\infty}$ norm) are also worth further investigation and development.

\section*{Funding Statement}
Research
reported in this publication was supported by the National Institute Of General Medical Sciences
of the National Institutes of Health under Award Number R01GM140476. The content is solely the
responsibility of the authors and does not necessarily represent the official views of the National
Institutes of Health. Portions of this research were conducted with the advanced computing resources provided by Texas A\&M High Performance Research Computing.

\section*{Data Availability Statement}
The data underlying this article are publicly available via Zenodo at \url{https://doi.org/10.5281/zenodo.2348892}.

\section*{Supplementary}
\subsection*{\texttt{R} code for ``Frequency Band Analysis of Multivariate Time Series".}
\texttt{R} code, a quick start demo, and descriptions of all functions and parameters needed to generate simulated data introduced in Section \ref{sec:simulations} and to implement the proposed method on data for use in practice can be downloaded from GitHub at this link: \url{https://github.com/sbruce23/mEBA}.

\bibliographystyle{chicago}
\bibliography{EBAbib}

\appendix
\section{Proofs}
\label{s:proofs}

\begin{proof}[\textbf{Proof of Theorem \ref{thm:consistency_test_statistic}}]
We have from \eqref{eq:main_discrepance_measure_estimated} the estimated discrepancy measure given by 
\begin{gather*}
\widehat{D}(\omega) = \frac{1}{T} \sum_{t=1}^T \frac{1}{W} \sum_{k=1}^{W} \sum_{a,b=1}^{p} \Big[ \widehat{g}_{a,b} \Big( \frac{t}{T} , \omega - \lambda_k \Big) - \widehat{g}_{a,b} \Big( \frac{t}{T} , \omega + \lambda_k \Big) \Big] \times \Big[ \widehat{g}_{a,b} \Big( \frac{t}{T} , \omega - \lambda_k \Big) - \\
\widehat{g}_{a,b} \Big( \frac{t}{T} , \omega + \lambda_k \Big) \Big]^{*} \\ \notag
 = \frac{1}{T} \sum_{t=1}^T \frac{1}{W} \sum_{k=1}^{W} \sum_{a,b=1}^{p} \Big[ I_{N,_{a,b}}( \frac{t}{T} , \omega - \lambda_k ) - \frac{1}{T} \sum_{t_1=1}^{T} I_{N,_{a,b}}( \frac{t_1}{T} , \omega - \lambda_k ) -  \\
  I_{N,_{a,b}}( \frac{t}{T} , \omega + \lambda_k ) + \frac{1}{T} \sum_{t_2=1}^{T} I_{N,_{a,b}}( \frac{t_2}{T} , \omega + \lambda_k ) \Big] \times \\ \notag
  \Big[ I_{N,_{a,b}}( \frac{t}{T} , \omega - \lambda_k ) - \frac{1}{T} \sum_{t_1=1}^{T} I_{N,_{a,b}}( \frac{t_1}{T} , \omega - \lambda_k ) -  I_{N,_{a,b}}( \frac{t}{T} , \omega + \lambda_k ) + \frac{1}{T} \sum_{t_2=1}^{T} I_{N,_{a,b}}( \frac{t_2}{T} , \omega + \lambda_k ) \Big]^{*}.
\end{gather*}
where $\widehat{g}_{a,b}$ and $I_{N,a,b}$ denote the entry $(a,b)$ in the respective matrices. Expanding the inside term of the above expression involves several terms  and we consider one of each kind and show that the expected value of the discrepancy measure tends to zero under $H_0$. 

First, we consider the terms of the type $I_N( \frac{t}{T} , \omega \pm \lambda_k )I_N^{*}( \frac{t}{T} , \omega \pm \lambda_k )$. We have, for component $(a,b)$,
\begin{gather*}
I_{N,_{a,b}}( \frac{t}{T} , \omega \pm \lambda_k )I_{N,_{a,b}}^{*}( \frac{t}{T} , \omega \pm \lambda_k ) = J_{N,_{a}}( \frac{t}{T} , \omega \pm \lambda_k )J_{N,_{b}}^{*}( \frac{t}{T} , \omega \pm \lambda_k )J_{N,_{a}}( \frac{t}{T} , \omega \pm \lambda_k )J_{N,_{b}}^{*}( \frac{t}{T} , \omega \pm \lambda_k ) \\ \notag
 = \frac{1}{(2 \pi N)^2} \Big( \sum_{s_1 = 0}^{N-1} X_{ a ,  \floor{u_t T} - N/2 + 1 + s_1,T } \; e^{-i s_1 \theta_{k,\pm}} \Big) \times \Big( \sum_{s_2 = 0}^{N-1} X_{ b ,  \floor{u_t T} - N/2 + 1 + s_2,T } \; e^{i s_2 \theta_{k,\pm}} \Big) \times  \\ \notag 
 \Big( \sum_{s_3 = 0}^{N-1} X_{ a , \floor{u_t T} - N/2 + 1 + s_3,T } \; e^{-i s_3 \theta_{k,\pm}} \Big) \times \Big( \sum_{s_4 = 0}^{N-1} X_{ b ,  \floor{u_t T} - N/2 + 1 + s_4,T } \; e^{i s_4 \theta_{k,\pm}} \Big) \\
  = \frac{1}{ (2 \pi N)^2} \sum_{s_1,s_2,s_3,s_4} \sum_{l,m,n,o = -\infty}^{\infty} \Big( \Phi_a( u_{t_{s_1}} , l) \varepsilon_{t_{s_1} - l} \Big) \times \Big( \Phi_b ( u_{t_{s_2}} , m) \varepsilon_{t_{s_2} - m} \Big) \\ 
  \Big( \Phi_a ( u_{t_{s_3}} , n) \varepsilon_{t_{s_3} - n} \Big) \times \Big( \Phi_b ( u_{t_{s_4}} , o) \varepsilon_{t_{s_4} - o} \Big) \textrm{exp}(-i \theta_{k,\pm} (s_1-s_2+s_3-s_4) ) \; + \; O(\frac{1}{T}) \\
 =  \frac{1}{ (2 \pi N)^2} \sum_{s_1,s_2,s_3,s_4} \sum_{l,m,n,o = -\infty}^{\infty} \Big( Z_{a,t_{s_1} - l}  Z_{b,t_{s_2} - m} Z_{a,t_{s_3} - n} Z_{b,t_{s_4} - o}  \Big)  \textrm{exp}(-i \theta_{k,\pm} (s_1-s_2+s_3-s_4) ) \; + \; O(\frac{1}{T})
\end{gather*}
where $\theta_{k,\pm} = \omega \pm \lambda_k$, $u_t = \frac{t}{T}$ and $t_{s_j} = \floor{u_t T} - N/2 + 1 + s_j$ for $j=1,2,3,4$. Also, $Z_{a,t_{s_1} - l} = \Big( \Phi_a^{'}( u_{t_{s_1}} , l) \varepsilon_{t_{s_1} - l} \Big)$ where $\Phi_a^{'}( u_{t_{s_1}} , l)$ denotes the $a^{th}$ row of the coefficient matrix $\Phi( u_{t_{s_1}} , l)$. For the expectation above, we apply Theorem 2.3.2 of \citet{brillinger81}. Noting that the $Z$ random variables above are Gaussian by the assumption in \eqref{eq:locally_stationary_ts}, the expected values simplifies to 

\begin{gather*}
E \Big( I_{N,_{a,b}}( \frac{t}{T} , \omega \pm \lambda_k )I_{N,_{a,b}}^{*}( \frac{t}{T} , \omega \pm \lambda_k ) \Big) = E_{1,T}^{(a,b)} + E_{2,T}^{(a,b)} +  O(\frac{1}{N}) + O(\frac{1}{T}) + O(\frac{N^2}{T^2}),
\end{gather*} 
where 
\begin{gather*}
 E_{1,T}^{(a,b)} = \frac{1}{ (2 \pi N)^2} \sum_{s_1,s_2,s_3,s_4} \sum_{l,m,n,o = -\infty}^{\infty} \textrm{exp}(-i \theta_{k,\pm} (s_1-s_2+s_3-s_4) )   E \Big[ \Big( \Phi_a^{'}( u_{t} , l) \varepsilon_{t_{s_1} - l} \Big) \times \Big( \Phi_b^{'}( u_{t} , m) \varepsilon_{t_{s_2} - m} \Big) \Big] \times  \\
 E \Big[ \Big( \Phi_a^{'}( u_{t} , n) \varepsilon_{t_{s_3} - n} \Big) \times \Big( \Phi_b^{'}( u_{t} , o) \varepsilon_{t_{s_4} - o} \Big) \Big]  \\
E_{2,T}^{(a,b)} = \frac{1}{ (2 \pi N)^2} \sum_{s_1,s_2,s_3,s_4} \sum_{l,m,n,o = -\infty}^{\infty} \textrm{exp}(-i \theta_{k,\pm} (s_1-s_2+s_3-s_4) )  E \Big[ \Big( \Phi_a^{'}( u_{t} , l) \varepsilon_{t_{s_1} - l} \Big) \times \Big( \Phi_b^{'}( u_{t} , o) \varepsilon_{t_{s_4} - o} \Big) \Big] \times \\
 E \Big[ \Big( \Phi_a^{'}( u_{t} , n) \varepsilon_{t_{s_3} - n} \Big) \times \Big( \Phi_b^{'}( u_{t} , m) \varepsilon_{t_{s_3} - m} \Big) \Big]
\end{gather*}

For $E_{1,T}^{(a,b)}$ it can be seen that the expectations are non-zero only when $t_{s_1}-l=t_{s_2}-m$ and $t_{s_3}-n=t_{s_4}-o$. We hence have   
\begin{gather} \label{eq:e1_term}
\frac{1}{T} \sum_{t=1}^T \frac{1}{W} \sum_{k=1}^{W} E_{1,T}^{(a,b)} = \frac{1}{T} \sum_{t=1}^T \frac{1}{W} \sum_{k=1}^{W} f_{a,b}(\frac{t}{T} , \theta_{k,\pm}) f_{a,b}(u_t , \theta_{k,\pm})^{*} + o(1). 
\end{gather}
Similarly, for  the term $E_{2,T}^{(a,b)}$ it can be seen that the expectations are non-zero only when $t_{s_1}-l=t_{s_4}-o$ and $t_{s_3}-n=t_{s_2}-m$. We get   
\begin{gather} \label{eq:eq:e2_term}
\frac{1}{T} \sum_{t=1}^T \frac{1}{W} \sum_{k=1}^{W} E_{2,T}^{(a,b)} = \frac{1}{T} \sum_{t=1}^T \frac{1}{W} \sum_{k=1}^{W} f_{a,b}(\frac{t}{T} , \theta_{k,\pm}) f_{a,b}(u_t , \theta_{k,\pm})^{*} + o(1). 
\end{gather}
where $\theta_{k,\pm} = \omega \pm \lambda_k$, $u_t = t/T$, $f_{a,b}(u_t , \theta_{k,\pm})$ denotes component $(a,b)$ of the spectral matrix $f(u_t , \theta_{k,\pm})$ and $f_{a,b}(u_t , \theta_{k,\pm})^{*}$ is the conjugate transpose. The terms in \eqref{eq:e1_term} and \eqref{eq:eq:e2_term} are approximated by $\frac{1}{2 \pi c}\int_0^1  \int_{0}^{2 \pi c} f_{a,b}(u,\omega \pm \lambda)f_{a,b}(u, \omega \pm \lambda)^{*} d \lambda \; du + o(1)$. 

Next, we consider the terms of the type $I_N( \frac{t}{T} , \omega \pm \lambda_k )I_N^{*}( \frac{t}{T} , \omega \mp \lambda_k )$. We have, for component $(a,b)$,
\begin{gather*}
I_{N,_{a,b}}( \frac{t}{T} , \omega \pm \lambda_k )I_{N,_{a,b}}^{*}( \frac{t}{T} , \omega \mp \lambda_k ) = 
  \frac{1}{(2 \pi N)^2} \Big( \sum_{s_1 = 0}^{N-1} X_{ a ,  \floor{u_t T} - N/2 + 1 + s_1,T } \; e^{-i s_1 \theta_{k,\pm}} \Big) \times \\ \notag
  \Big( \sum_{s_2 = 0}^{N-1} X_{ b ,  \floor{u_t T} - N/2 + 1 + s_2,T } \; e^{i s_2 \theta_{k,\pm}} \Big) \times  
 \Big( \sum_{s_3 = 0}^{N-1} X_{ a , \floor{u_t T} - N/2 + 1 + s_3,T } \; e^{-i s_3 \theta_{k,\mp}} \Big) \times \Big( \sum_{s_4 = 0}^{N-1} X_{ b ,  \floor{u_t T} - N/2 + 1 + s_4,T } \; e^{i s_4 \theta_{k,\mp}} \Big) \\
 =  \frac{1}{ (2 \pi N)^2} \sum_{s_1,s_2,s_3,s_4} \sum_{l,m,n,o = -\infty}^{\infty} \Big( Z_{a,t_{s_1} - l}  Z_{b,t_{s_2} - m} Z_{a,t_{s_3} - n} Z_{b,t_{s_4} - o}  \Big)  \textrm{exp}(-i \theta_{k,\pm} (s_1-s_2) ) \textrm{exp}(-i \theta_{k,\mp} (s_3-s_4) ) \; +  \\ \notag
 \; O(\frac{1}{T}),
\end{gather*}
where $\theta_{k,\pm} = \omega \pm \lambda_k$, $\theta_{k,\mp} = \omega \mp \lambda_k$. The expected values simplifies to  $E \Big( I_{N,_{a,b}}( \frac{t}{T} , \omega \pm \lambda_k )I_{N,_{a,b}}^{*}( \frac{t}{T} , \omega \pm \lambda_k ) \Big) = E_{3,T}^{(a,b)} + E_{4,T}^{(a,b)} +  o(1)$, where 
\begin{gather*}
 E_{3,T}^{(a,b)} = \frac{1}{ (2 \pi N)^2} \sum_{s_1,s_2,s_3,s_4} \sum_{l,m,n,o = -\infty}^{\infty} \textrm{exp}(-i \theta_{k,\pm} (s_1-s_2) ) \textrm{exp}(-i \theta_{k,\mp} (s_3-s_4) )   E \Big[ \Big( \Phi_a^{'}( u_{t} , l) \varepsilon_{t_{s_1} - l} \Big) \\ \notag
 \times \Big( \Phi_b^{'}( u_{t} , m) \varepsilon_{t_{s_2} - m} \Big) \Big] \times 
 E \Big[ \Big( \Phi_a^{'}( u_{t} , n) \varepsilon_{t_{s_3} - n} \Big) \times \Big( \Phi_b^{'}( u_{t} , o) \varepsilon_{t_{s_4} - o} \Big) \Big]  \\
E_{4,T}^{(a,b)} = \frac{1}{ (2 \pi N)^2} \sum_{s_1,s_2,s_3,s_4} \sum_{l,m,n,o = -\infty}^{\infty} \textrm{exp}(-i \theta_{k,\pm} (s_1-s_4) ) \textrm{exp}(-i \theta_{k,\mp} (s_3-s_2) )  E \Big[ \Big( \Phi_a^{'}( u_{t} , l) \varepsilon_{t_{s_1} - l} \Big) \\ \notag
\times \Big( \Phi_b^{'}( u_{t} , o) \varepsilon_{t_{s_4} - o} \Big) \Big] \times 
 E \Big[ \Big( \Phi_a^{'}( u_{t} , n) \varepsilon_{t_{s_3} - n} \Big) \times \Big( \Phi_b^{'}( u_{t} , m) \varepsilon_{t_{s_3} - m} \Big) \Big]. 
\end{gather*} 

For $E_{3,T}^{(a,b)}$, the expectations are non-zero only when $t_{s_1}-l=t_{s_2}-m$ and $t_{s_3}-n=t_{s_4}-o$. We hence have   
\begin{gather} \label{eq:e3_term}
\frac{1}{T} \sum_{t=1}^T \frac{1}{W} \sum_{k=1}^{W} E_{1,T}^{(a,b)} = \frac{1}{T} \sum_{t=1}^T \frac{1}{W} \sum_{k=1}^{W} f_{a,b}(\frac{t}{T} , \theta_{k,\pm}) f_{a,b}(u_t , \theta_{k,\mp})^{*} + o(1). 
\end{gather}
For  $E_{4,T}^{(a,b)}$, the expectations are non-zero only when $t_{s_1}-l=t_{s_4}-o$ and $t_{s_3}-n=t_{s_2}-m$. We get   
\begin{gather} \label{eq:eq:e4_term}
\frac{1}{T} \sum_{t=1}^T \frac{1}{W} \sum_{k=1}^{W} E_{2,T}^{(a,b)} = \frac{1}{T} \sum_{t=1}^T \frac{1}{W} \sum_{k=1}^{W} f_{a,b}(\frac{t}{T} , \theta_{k,\pm}) f_{a,b}(u_t , \theta_{k,\mp})^{*} + o(1). 
\end{gather}

The remaining types of terms considered are i). $I_N( \frac{t}{T} , \omega \pm \lambda_k ) \Big( \frac{1}{T} \sum_{t_1=1}^{T} I_{N,_{a,b}}( \frac{t_1}{T} , \omega \pm \lambda_k )^{*} \Big)$; ii). $I_N( \frac{t}{T} , \omega \pm \lambda_k ) \Big( \frac{1}{T} \sum_{t_1=1}^{T} I_{N,_{a,b}}( \frac{t_1}{T} , \omega \mp \lambda_k )^{*} \Big)$; iii). $ \Big( \frac{1}{T} \sum_{t_1=1}^{T} I_{N,_{a,b}}( \frac{t_1}{T} , \omega \pm \lambda_k )^{*} \Big) I_N( \frac{t}{T} , \omega \pm \lambda_k )^{*}$; iv). $ \Big( \frac{1}{T} \sum_{t_1=1}^{T} I_{N,_{a,b}}( \frac{t_1}{T} , \omega \pm \lambda_k )^{*} \Big) I_N( \frac{t}{T} , \omega \mp \lambda_k )^{*}$; v). $ \Big( \frac{1}{T} \sum_{t_1=1}^{T} I_{N,_{a,b}}( \frac{t_1}{T} , \omega \pm \lambda_k ) \Big) \times \Big( \frac{1}{T} \sum_{t_1=1}^{T} I_{N,_{a,b}}( \frac{t_1}{T} , \omega \pm \lambda_k )^{*} \Big) $ and vi).  $ \Big( \frac{1}{T} \sum_{t_1=1}^{T} I_{N,_{a,b}}( \frac{t_1}{T} , \omega \pm \lambda_k ) \Big) \times \Big( \frac{1}{T} \sum_{t_1=1}^{T} I_{N,_{a,b}}( \frac{t_1}{T} , \omega \mp \lambda_k )^{*} \Big) $. The treatment of all these types of terms is similar to the approach that lead to \eqref{eq:e1_term}-\eqref{eq:eq:e4_term} for the first two types of terms considered.

When $\omega \in C_{N,2}$, it can be seen that by combining the 16 limiting expressions that we get from all the different types of terms, the expected value of $\widehat{D}(\omega)$ tends to zero. With the same approach, when $\omega \in \{ \omega_1 , \omega_2 , \hdots , \omega_K \}$, the expected value of $\widehat{D}(\omega)$ is approximated by  $\frac{1}{2 \pi c} \int_{0}^1 \int_0^{2 \pi c} \sum_{a,b=1}^p \Big( g_{a,b}(u,\omega - \lambda) - g_{a,b}(u,\omega + \lambda) \Big) \Big( g_{a,b}(u,\omega - \lambda) - g_{a,b}(u,\omega + \lambda) \Big)^{*} \; d \lambda \; du + o(1)$.

Next, we consider the variance of the discrepancy measure. Here we look at the $2^{nd}$ order cumulant of discrepancy measure that is written as
\begin{gather} \label{eq:var_exp_main}
cum(\widehat{D}(\omega) ,\widehat{D}(\omega)) = \frac{1}{T^2} \sum_{t_1,t_2=1}^T \frac{1}{W^2} \sum_{k_1,k_2=1}^{W} \sum_{a,b,c,d=1}^{p}  cum \Bigg( \Big[ I_{N,_{a,b}}( \frac{t_1}{T} , \omega - \lambda_{k_1} ) - \frac{1}{T} \sum_{x_1=1}^{T} I_{N,_{a,b}}( \frac{x_1}{T} , \omega - \lambda_{k_1} ) -  \notag \\
  I_{N,_{a,b}}( \frac{t_1}{T} , \omega + \lambda_{k_1} ) + \frac{1}{T} \sum_{x_2=1}^{T} I_{N,_{a,b}}( \frac{x_2}{T} , \omega + \lambda_{k_1} ) \Big] \times \notag \\ 
  \Big[ I_{N,_{a,b}}( \frac{t_1}{T} , \omega - \lambda_{k_1} ) - \frac{1}{T} \sum_{x_3=1}^{T} I_{N,_{a,b}}( \frac{x_3}{T} , \omega - \lambda_{k_1} ) -  I_{N,_{a,b}}( \frac{t_1}{T} , \omega + \lambda_{k_1} ) + \frac{1}{T} \sum_{x_4=1}^{T} I_{N,_{a,b}}( \frac{x_4}{T} , \omega + \lambda_{k_1} ) \Big]^{*} , \notag \\
  \Big[ I_{N,_{c,d}}( \frac{t_2}{T} , \omega - \lambda_{k_2} ) - \frac{1}{T} \sum_{x_1=1}^{T} I_{N,_{c,d}}( \frac{x_1}{T} , \omega - \lambda_{k_2} ) -  \notag \\
  I_{N,_{c,d}}( \frac{t_2}{T} , \omega + \lambda_{k_2} ) + \frac{1}{T} \sum_{x_2=1}^{T} I_{N,_{c,d}}( \frac{x_2}{T} , \omega + \lambda_{k_2} ) \Big] \times  \notag \\ 
  \Big[ I_{N,_{c,d}}( \frac{t_2}{T} , \omega - \lambda_{k_2} ) - \frac{1}{T} \sum_{x_3=1}^{T} I_{N,_{c,d}}( \frac{x_3}{T} , \omega - \lambda_{k_2} ) -  I_{N,_{c,d}}( \frac{t_2}{T} , \omega + \lambda_{k_2} ) + \frac{1}{T} \sum_{x_4=1}^{T} I_{N,_{c,d}}( \frac{x_4}{T} , \omega + \lambda_{k_2} ) \Big]^{*} \Bigg).\end{gather} 

We can now look at cumulant terms of different types and see the behavior as $T \rightarrow \infty$. First, for components $(a,b)$ and $(c,d)$,we consider cumulant terms of the type  

\begin{gather} \label{eq:var_term_1}
cum \Big( I_{N,_{a,b}}( \frac{t_1}{T} , \omega \pm \lambda_{k_1} )I_{N,_{a,b}}^{*}( \frac{t_1}{T} , \omega \pm \lambda_{k_1} ) , I_{N,_{c,d}}( \frac{t_2}{T} , \omega \pm \lambda_{k_2} )I_{N,_{c,d}}^{*}( \frac{t_2}{T} , \omega \pm \lambda_{k_2} ) \Big) =  \notag \\ 
\frac{1}{ (2 \pi N)^4} \sum_{r_1,r_2,r_3,r_4 = 0}^{N-1} \sum_{s_1,s_2,s_3,s_4 = 0}^{N-1}   \sum_{l_1,m_1,n_1,o_1 = -\infty}^{\infty} \sum_{l_2,m_2,n_2,o_2 = -\infty}^{\infty} \notag \\
cum \Bigg( \Big( \Phi_a( u_{t_{r_1}} , l) \varepsilon_{t_{r_1} - l_1} \Big) \times \Big( \Phi_b ( u_{t_{r_2}} , m_1) \varepsilon_{t_{r_2} - m_1} \Big) \times
  \Big( \Phi_a ( u_{t_{r_3}} , n_1) \varepsilon_{t_{r_3} - n_1} \Big) \times \Big( \Phi_b ( u_{t_{r_4}} , o_1) \varepsilon_{t_{r_4} - o_1} \Big) ,  \notag  \\ 
  \Big( \Phi_c( u_{t_{s_1}} , l_2) \varepsilon_{t_{s_1} - l_2} \Big) \times \Big( \Phi_d ( u_{t_{s_2}} , m_2) \varepsilon_{t_{s_2} - m_2} \Big) \times
  \Big( \Phi_c ( u_{t_{s_3}} , n_2) \varepsilon_{t_{s_3} - n_2} \Big) \times \Big( \Phi_d ( u_{t_{s_4}} , o_2) \varepsilon_{t_{s_4} - o_2} \Big) \Bigg) \times \notag \\ 
  \textrm{exp}(-i \theta_{k_1} (r_1-r_2+r_3-r_4) ) \times  \textrm{exp}(-i \theta_{k_2} (s_1-s_2+s_3-s_4) ) + O(\frac{1}{T}),
\end{gather}
where $\theta_{k_1,\pm} = \omega \pm \lambda_{k_1}$ and $\theta_{k_2,\pm} = \omega \pm \lambda_{k_2}$. Noting that $\varepsilon_t$ is Gaussian, an application of Theorem 2.3.2 of \citet{brillinger81} yield certain terms that are non-vanishing asymptotically. The term in \eqref{eq:var_term_1} leads to 
\begin{gather}
\frac{1}{T^2} \sum_{t_1,t_2=1}^T \frac{1}{W^2} \sum_{k_1,k_2=1}^{W} cum \Big( I_{N,_{a,b}}( \frac{t_1}{T} , \theta_{k_1,\pm} )I_{N,_{a,b}}^{*}( \frac{t_1}{T} , \theta_{k_1,\pm} ) , I_{N,_{c,d}}( \frac{t_2}{T} , \theta_{k_2,\pm} )I_{N,_{c,d}}^{*}( \frac{t_2}{T} , \theta_{k_2,\pm} ) \Big) =   \notag \\
\frac{1}{T^2} \sum_{t=1}^T \frac{1}{W^2} \sum_{k_1,k_2=1}^{W} \frac{40}{ (2 \pi)^4 } f_{a,b}(\frac{t}{T} , \theta_{k_1,\pm})f_{a,b}(\frac{t}{T} , \theta_{k_1,\pm})^{*} f_{c,d}(\frac{t}{T} , \theta_{k_2,\pm})f_{c,d}(\frac{t}{T} , \theta_{k_2,\pm})^{*} + o(1). 
\end{gather}  
The above term can be approximated by $\frac{1}{T 2 \pi c}\int_0^1  \int_{0}^{2 \pi c} f_{a,b}(u,\omega \pm \lambda)f_{a,b}(u, \omega \pm \lambda)^{*} f_{c,d}(u,\omega \pm \lambda)f_{c,d}(u, \omega \pm \lambda)^{*} d \lambda \; du + o(1)$. The treatment of the cumulant terms from the other term types follows similarly. 
\end{proof}

\begin{proof}[\textbf{Proof of Theorem \ref{thm:consistency_number_of_cpts}}]

With the sets $C_{N,1}$ and $C_{N,2}$ defined in \eqref{eq:points_sets_defn}, the proof of this result follows from the following two results.

\begin{itemize}

\item[(a).] $P \Big( \bigcup_{\omega \in C_{N,2}} \widehat{D}(\omega) > \eta_{T,\omega}  \Big) \overset{T \rightarrow \infty}{\longrightarrow } 0  $, 

\item[(b).] $P \Big( \bigcap_{\omega \in \{\omega_1,\omega_2,\hdots,\omega_K \} }  \widehat{D}(\omega) > \eta_{T,\omega}  \Big) \overset{T \rightarrow \infty}{\longrightarrow } 1$. 
\end{itemize}

For any $\omega \in C_{N,2}$, an application of Theorem \ref{thm:consistency_test_statistic}(a) yields

\begin{gather}
P( \bigcup_{\omega \in C_{N,2}} \widehat{D}(\omega) > \eta_{T,\omega}  ) \leq \sum_{\omega \in C_{N,2}} P(  \widehat{D}(\omega) > \eta_{T,\omega}  ) \\
\leq  \sum_{\omega \in C_{N,2}}  \frac{E(\widehat{D}(\omega))}{\eta_{T,\omega}}  \; \overset{T \rightarrow \infty}{\longrightarrow } 0. 
\end{gather} 

The result in (b) follows by an application of Theorem \ref{thm:consistency_test_statistic}(b).

\end{proof}

\end{document}